\documentclass[11pt]{amsart}
\usepackage{amssymb,mathrsfs,graphicx,enumerate}
\usepackage{amsmath,amsfonts,amssymb,amscd,amsthm,bbm}
\usepackage[retainorgcmds]{IEEEtrantools}
\usepackage{colortbl}

\title[Emergent dynamics of the Lohe hermitian sphere model]{From the Lohe tensor model to the Lohe hermitian sphere model and emergent dynamics}

\author[Ha]{Seung-Yeal Ha}
\address[Seung-Yeal Ha]{\newline Department of Mathematical Sciences\newline Seoul National University, Seoul 08826, and \newline
Korea Institue for Advanced Study, Hoegiro 85, 02455, Seoul, Korea (Republic of)}
\email{syha@snu.ac.kr}

\author[Park]{Hansol Park}
\address[Hansol Park]{\newline Department of Mathematical Sciences\newline Seoul National University, Seoul 08826, Korea (Republic of)}
\email{hansol960612@snu.ac.kr}

\newtheorem{theorem}{Theorem}[section]
\newtheorem{lemma}{Lemma}[section]
\newtheorem{corollary}{Corollary}[section]
\newtheorem{proposition}{Proposition}[section]
\newtheorem{remark}{Remark}[section]

\newtheorem{definition}{Definition}[section]

\newcommand{\bbr}{\mathbb R}
\newcommand{\bbu}{\mathbb U}

\newcommand{\bbs}{\mathbb S}

\newcommand{\bbc}{\mathbb C}

\begin{document}

\date{\today}


\thanks{\textbf{Acknowledgment.} The work of S.-Y. Ha is supported by NRF-2020R1A2C3A01003881}

\begin{abstract}
We study emergent behaviors of the Lohe hermitian sphere(LHS) model which is an aggregation model on ${\mathbb C}^d$. The LHS model is a complex analog of the Lohe sphere model on ${\mathbb R}^d$, and hermitian spheres are invariant sets for the LHS dynamics. For the derivation of the LHS model, we use a top-down approach, namely a reduction from a high-rank aggregation model, ``{\it the Lohe tensor model}."   The Lohe tensor model is a first-order aggregation model on the space of tensors with the same rank and sizes, and it was  first proposed by the authors in a recent work \cite{H-P}. In this work, we study how the LHS model appears as a special case of the Lohe tensor model and for the proposed model, we provide a cross-ratio like conserved quantity, a sufficient framework for the complete aggregation and a uniform $\ell^p$-stability estimate with respect to initial data.
\end{abstract}
\maketitle \centerline{\date}


\section{Introduction} \label{sec:1}
\setcounter{equation}{0} 
Collective behaviors of complex systems often appear in biological and physical systems, for example, {aggregation of bacteria \cite{T-B-L, T-B}, flocking of birds \cite{B-H}, swarming of fish \cite{B-T1, B-T2}, flashing of fireflies \cite{A-B, B-B, Wi1}, synchronous heart beating \cite{Pe}, an array of Josephson junctions  etc. We refer to \cite{A-B, A-B-F, H-K-P-Z, P-R, St,VZ} for survey articles on the collective dynamics.} Despite of its ubiquity in our nature, rigorous  mathematical analysis for such collective behaviors was begun only in a half century ago by Winfree \cite{Wi2, Wi1} and Kuramoto \cite{Ku1, Ku2}. Since then, several particle or agent-based models were proposed and studied in several disciplines such as applied mathematics, control theory, nonlinear dynamics and statistical physics, etc. In this work, our main interest lies on the Lohe tensor model \cite{H-P} which was proposed for the aggregate modeling of the ensemble of tensors. The Lohe tensor model is a natural extension of the well-known aggregation models such as {the Kuramoto model \cite{B-C-M, C-H-J-K, C-S, D-X,D-B1,D-B0, D-B, H-K-R, H-L-X, V-M1, V-M2}, the Lohe sphere model \cite{C-C-H, C-H1,C-H2,C-H3, C-H4, C-H5,H-K-P-R, H-K-R, H-K-RY, J-C, M-T-G, T-M, Z-Z, Z-Z-Q, Zhu}, the Schr\"{o}dinger-Lohe model \cite{C-H4}, and other matrix models in \cite{B-C-S, D-F-M-T, D-F-M, De, H-K, H-R} (see Section \ref{sec:2} for their relations). }To be more concrete, we begin with a crash introduction on tensors.

Let ${\mathbb C}^d$ and ${\mathbb C}^{d_1 \times d_2}$ be a $d$-dimensional complex vector space isomorphic to $\bbr^{2d}$ and a complex vector space consisting of $d_1 \times d_2$ matrices with complex entries, respectively. A rank-$m$ tensor is defined as a complex-valued multi-linear map from $\bbc^{d_1\times\cdots\times d_m}$ to $\bbc$. Intuitively, a tensor can be visualized as an $m$-dimensional array of complex numbers with $m$-indices. The rank of a tensor is the number of indices, say a rank-$m$ tensor of dimensions $d_1 \times \cdots \times d_m$ is an element of ${\mathbb C}^{d_1 \times \cdots \times d_m}$.  It is easy to see that scalars, vectors and matrices correspond to rank-0, 1 and 2 tensors, respectively.  Let $T$ be a rank-$m$ tensor. Then, we denote $(\alpha_1, \cdots, \alpha_m)$-th component of the tensor $T$ by $[T]_{\alpha_1 \cdots \alpha_m}$, and we also denote $\bar{T}$ by the rank-$m$ tensor whose components are the complex conjugate of the elements in $T$:
\[ [\bar{T}]_{\alpha_1 \cdots \alpha_m} =\overline{[T]_{\alpha_1 \cdots \alpha_m}}. \]
Let ${\mathcal T}_m(\bbc; d_1 \times\cdots\times d_m)$ be the set of all rank-$m$ tensors with size $d_1 \times\cdots\times d_m$. Then, the Kuramoto model, the Lohe sphere model and the Lohe matrix model can be regarded as aggregation models on ${\mathcal T}_0(\bbr; 0), {\mathcal T}_1(\bbr; d)$ and ${\mathcal T}_2(\bbc; d \times d)$, respectively. Let $\{T_j \}_{j=1}^{N}$ be the $N$-collection of rank-$m$ tensors in ${\mathcal T}_m(\bbc; d_1 \times\cdots\times d_m)$, and $A_j$ is the skew-hermitian rank-$2m$ tensors with size $(d_1 \times\cdots\times d_m) \times (d_1 \times \cdots\times d_m)$. For the simplicity of presentation, we introduce handy notation as follows: for $T \in {\mathcal T}_m(\bbc; d_1 \times \cdots\times d_m)$ and $A \in  {\mathcal T}_{2m}(\bbc; d_1 \times\cdots\times  d_m \times d_1 \times \cdots\times d_m)$, we set
\begin{align*}
\begin{aligned}
& [T]_{\alpha_{*}}:=[T]_{\alpha_{1}\alpha_{2}\cdots\alpha_{m}}, \quad [T]_{\alpha_{*0}}:=[T]_{\alpha_{10}\alpha_{20}\cdots\alpha_{m0}},  \quad  [T]_{\alpha_{*1}}:=[T]_{\alpha_{11}\alpha_{21}\cdots\alpha_{m1}}, \\
&  [T]_{\alpha_{*i_*}}:=[T]_{\alpha_{1i_1}\alpha_{2i_2}\cdots\alpha_{mi_m}}, \quad [T]_{\alpha_{*(1-i_*)}}:=[T]_{\alpha_{1(1-i_1)}\alpha_{2(1-i_2)}\cdots\alpha_{m(1-i_m)}}, \\
&  [A]_{\alpha_*\beta_*}:=[A]_{\alpha_{1}\alpha_{2}\cdots\alpha_{m}\beta_1\beta_2\cdots\beta_{m}},
\end{aligned}
\end{align*}
and we set the average of $T_i$:
{{\[ T_c := \frac{1}{N} \sum_{k=1}^{N} T_k. \]}}
Then, the Lohe tensor model \cite{H-P} in component-wise form can be written as follows:
\begin{equation}
\begin{cases} \label{M-1}
\displaystyle \dot{[T_j]}_{\alpha_{*0}} = [A_j]_{\alpha_{*0}\alpha_{*1}}[T_j]_{\alpha_{*1}} \\
\displaystyle \hspace{0.5cm} + \sum_{i_* \in \{0, 1\}^m}\kappa_{i_*} \Big([T_c]_{\alpha_{*i_*}}\bar{[T_j]}_{\alpha_{*1}}[T_j]_{\alpha_{*(1-i_*)}}-[T_j]_{\alpha_{*i_*}}\bar{[T_c]}_{\alpha_{*1}}[T_j]_{\alpha_{*(1-i_*)}} \Big), \quad t>0,\\
T_j(0)=T_j^{in},\quad \|T_j^{in} \|_F = 1, \quad j=1, 2, \cdots, N, \\
\displaystyle  \bar{[A_j]}_{\alpha_{*0}\alpha_{*1}}=-[A_j]_{\alpha_{*1}\alpha_{*0}},
\end{cases}
\end{equation}
where $\kappa_{i_*}$'s are nonnegative coupling strengths, and $\| \cdot \|_F$ is the Frobenius norm defined in Definition \ref{D2.1}.

{Note that the first term and the second term in $\eqref{M-1}_1$ represent a free rotational flow, and cubic aggregation coupling for collective behaviors, respectively. The relation $\eqref{M-1}_3$ denotes the skew-hermitian property of the $2m$-tensor $A_j$. The cubic coupling terms in \eqref{M-1} were designed to generalize earlier aggregation models such as the Lohe sphere model in $\bbr^d$ and Lohe matrix model on the unitary group $\bbu(d)$. See Section \ref{sec:2.1} for motivation. The emergent dynamics of \eqref{M-1} has been studied in \cite{H-P} under a rather restricted framework.} Although the coupling terms in the R.H.S. of $\eqref{M-1}_1$ look complicated, system \eqref{M-1} has a natural conserved quantity and exhibits an emergent aggregation dynamics. For the choices:
\[ T_j = z_j \in {\mathcal T}_1(\bbc, d) = \bbc^d, \quad \mbox{and} \quad A_j = \Omega_j \in {\mathcal T}_2(\bbc, d \times d), \] 
system \eqref{M-1} reduces to the Lohe {hermitian sphere} model on $\bbc^d$:
\begin{equation}\label{M-5}
\begin{cases}
\dot{z}_j=\Omega_j z_j+\kappa_0 \Big (\langle{z_j, z_j}\rangle z_c-\langle{z_c, z_j}\rangle z_j \Big )+\kappa_1 \Big (\langle{z_j, z_c}\rangle-\langle{z_c, z_j}\rangle \Big )z_j, \quad t > 0, \\
z_j(0)=z_j^{in}, \quad  \|z^{in}_j \| = 1,
\end{cases}
\end{equation}
where {{$z_c :=\frac{1}{N}\sum_{k=1}^Nz_k$}}, $\langle z, w \rangle$ and $\Omega_j$ are standard inner product in $\bbc^d$ and a skew-hermitian $d \times d$ matrix satisfying
\[
\langle z, w \rangle = {\bar z}^{\alpha} w^{\alpha}, \quad  \Omega_j^* = -\Omega_j.
\]
Here we used Einstein summation convention. \newline

Note that for a real rank-1 tensor $z_j \in \bbr^d$, the second coupling terms involving with $\kappa_1$ vanishes, and we recover the Lohe sphere model on $\bbs^d$:
\begin{equation}\label{M-5-1}
\dot{x}_j=\Omega_j x_j+\kappa_0 \Big (\langle{x_j, x_j}\rangle x_c-\langle{x_c, x_j}\rangle x_j \Big ).
\end{equation}
{The emergent dynamics of \eqref{M-5-1} has been extensively studied from diverse aspects, e.g, complete aggregation under attractive couplings \cite{C-C-H,C-H1, J-C, Lo-1, Lo-2, M-T-G, Zhu}, interplay between attractive and repulsive couplings \cite{C-H5}, time-delay \cite{C-H2,C-H3}}. {Here ``{\it complete aggregation}" means that all relative states $z_i - z_j,~i \not = j$ tend to zero asymptotically (see Definition \ref{D2.1}).} \newline

In this paper, we only consider a homogeneous ensemble with $\Omega_j = \Omega,~j= 1, \cdots, N$, and due to the solution splitting property in Proposition \ref{P2.1}, we may assume $\Omega = 0$ without loss of generality. In this paper, we address the following issues for \eqref{M-5}:  
\begin{itemize}
\item
$({\mathcal Q}1)$:~Are there any nontrivial conserved quantties?

\vspace{0.2cm}

\item
$({\mathcal Q}2)$:~Under what conditions characterized in terms of  systems parameters and initial data, does system \eqref{M-5} exhibit emergent dynamics?
\end{itemize}
In this work, we deal with the questions $({\mathcal Q}1)$ and $({\mathcal Q}2)$ {which generalize the emergent dynamics in \cite{H-P} in which the condition $\kappa_0 \gg \kappa_1$ is required to guarantee emergent behaviors.} \newline

\noindent The main results of this paper are three-fold. First, as a first step for the emergent dynamics of full system \eqref{M-5}, we consider two Cauchy problems for sub-systems:
{{
\begin{align}
\begin{aligned} \label{M-6}
& \begin{cases}
\dot{z}_j= \kappa_0(\langle{z_j, z_j}\rangle z_c-\langle{z_c, z_j}\rangle z_j), \quad t > 0, \\
z_j(0) = z_j^{in}, \quad  \|z^{in}_j \| = 1,
\end{cases} \\
& \mbox{and} \\
& \begin{cases}
 \dot{z}_j=\kappa_1(\langle{z_j, z_c}\rangle-\langle{z_c, z_j}\rangle)z_j, \quad t > 0, \\
 z_j(0) = z_j^{in}, \quad  \|z^{in}_j \| = 1.
 \end{cases}
 \end{aligned}
\end{align}}}
{If we set $\kappa_1=0$, we obtain the first system of \eqref{M-6}, in contrast if we set $\kappa_0=0$, then we obtain the second equation of \eqref{M-6}.} From now on, we call the first and second systems as Subsystem A and Subsystem B, respectively. First, for the Subsystem A, we provide a nontrivial conserved quantity and a sufficient framework for complete aggregation (see Definition \ref{D2.1}). \newline

\noindent For a non-overlapping configuration $Z = (z_1, \cdots, z_N)$ with
\[ z_i \not = z_j, \quad  1 \leq i \not = j \leq N, \]
we introduce a cross-ratio like functional:
\begin{equation} \label{M-8}
 \mathcal{C}_{ijkl} :=\frac{(1-\langle{z_i, z_j}\rangle)(1-\langle{z_k, z_l}\rangle)}{(1-\langle{z_i, z_l}\rangle)(1-\langle{z_k, z_j}\rangle)}.
\end{equation}
Then, we can show that $\mathcal{C}_{ijkl}$ is a conserved quantity along the flow \eqref{M-6} (Proposition \ref{P3.1}):
\[
 \mathcal{C}_{ijkl}(t) =  \mathcal{C}_{ijkl}(0), \quad t \geq 0.
\]
Recently, the above cross ratio like quantity \eqref{M-8} was also introduced for the Lohe matrix model \cite{Lo-0} and Lohe sphere model in \cite{H-K-P-R}. For the emergent dynamics, we show that $h_{ij} := \langle z_i, z_j \rangle$ tends to one exponentially fast (see Theorem \ref{T3.1}), when the coupling strength, the initial data $\{z_j \}$ and coupling strength satisfy
\[
\kappa_0 > 0, \quad \|z_i^{in} \|=1,\quad \max_{i\neq j}|1-\langle{z_i^{in}, z_j^{in}}\rangle|<1/2.
\]
Second, we show that Subsystem B is completely determined by the phase dynamics of the Kuramoto model with frustration. More precisely, we verify that the general solution $z_j$ takes the form of 
\[ z_j(t)= e^{{\mathrm i} \theta_j(t)} z^{in}_j, \quad j = 1, \cdots, N, \]
where the dynamics for $\theta_j$ is governed by the Kuramoto model with frustration (see Theorem \ref{T3.2}):
\begin{equation*} 
\begin{cases}
\displaystyle \dot{\theta}_j=\frac{2 \kappa_1}{N}\sum_{k=1}^N R_{jk}^{in} \sin(\theta_k-\theta_j+\alpha_{jk}), \quad t > 0, \\
\displaystyle \theta_j(0)=0,
\end{cases}
\end{equation*}
where the amplitude $R_{jk}^{in}$ and frustration $\alpha_{jk}$ are completely determined by the initial data:
\[ \langle{z_j^{in}, z_k^{in}}\rangle=R^{in}_{jk}e^{\mathrm{i} \alpha_{jk}}. \]
Note that the emergent dynamics for the Kuramoto model has been extensively studied in literature, to name a few \cite{A-R,B-C-M, C-H-J-K,  C-S, D-X, D-B1,D-B0, D-B, H-K-R, H-L-X, M-S1, M-S2, M-S3, V-M1, V-M2}. {Details will be treated in Section 3.} 

Third, we consider an emergent dynamics to system \eqref{M-5} with $\Omega_j  = 0$ {to see the interplay of two competing mechanisms in Subsystem A and Subsystem B.} For this, we set
\begin{equation} \label{M-9}
\rho = \| z_c \|. 
\end{equation} 
Then the quantity $\rho$ has been used as an order parameter measuring the aggregation for the Lohe sphere model. In Lemma \ref{L4.1}, we derive a differential inequality for $\rho$:
\[
\frac{d\rho^2}{dt}=\frac{2\kappa_0}{N}\sum_{i = 1}^{N} \Big( \rho^2 -|\langle{z_i, z_c}\rangle|^2 \Big) +\frac{4(\kappa_0+\kappa_1)}{N}\sum_{i =1}^{N} \Big| \mathrm{Im}(\langle{z_i, z_c}\rangle) \Big|^2.
\]
For the complete aggregation of \eqref{M-5}, we provide a sufficient framework leading to the complete aggregation:
\[ 0<  \kappa_1 <  \frac{1}{4} \kappa_0.  \]
As long as the solution $\{z_j \}$ satisfies a priori condition:
\[ z_c(t) \neq 0, \quad t \geq 0, \]
the configuration diameter $\mathcal{D}(Z) := \max_{i,j} \|z_i - z_j \|$ converges to zero exponentially fast (see Theorem \ref{T4.1}). We also provide a uniform stability estimate of \eqref{M-5} with respect to initial data. More precisely, let $Z:= \{ z_j \}$ and ${\tilde Z} = \{ {\tilde z}_j \}$ be two global solutions to \eqref{M-5}. Then, we show that  there exists a positive constant $G$ independent of time $t$ such that 
\[
\sup_{0 \leq t < \infty} \|Z(t)-\tilde{Z}(t)\|_{p} \leq G\|Z^{in}-\tilde{Z}^{in}\|_{p},
\]
where $p \in [1, \infty)$ and $\|Z\|_p:=\left(\sum_{j=1}^N|z_j|^p\right)^{1/p}$. {Details can be found in Section \ref{sec:4}}. \newline

The rest of this paper is organized as follows. In Section \ref{sec:2}, we briefly discuss the modeling spirit of the Lohe tensor model and basic properties of the Lohe hermitian sphere model. In Section \ref{sec:3}, we study emergent properties of two subsystems for the LHS model. In Section \ref{sec:4}, we study emergent dynamics of the full LHS model for a homogeneous ensemble. Finally, Section \ref{sec:5} is devoted to a brief summary of our main results and some issues to be explored in a future work.  

\vspace{0.5cm}

\noindent {\bf Notation}: Throughout the paper, we use subscript and superscript to denote the number of particle and component, respectively. We set $z_j = (z_j^1, \cdots, z_j^d) \in \bbc^d$.

\section{Preliminaries} \label{sec:2}
\setcounter{equation}{0} 
In this section, we briefly review basic properties and emergent dynamics of the Lohe tensor model and the LHS model. First, we recall definitions on the Frobenius norm of a tensor and collective dynamics as follows.
\begin{definition} \label{D2.1}
Let $\{T_i \}$ be an ensemble of rank-$m$ tensors with the same size. 
\begin{enumerate}
\item
The Frobenius norm $\|\cdot\|_F$ is defined as follows:
\[
\|T\|_F :=\left(\sum_{\alpha_1, \alpha_2, \cdots, \alpha_m}\big|[T]_{\alpha_1\alpha_2\cdots\alpha_m}\big|^2\right)^{1/2}.
\]

\vspace{0.2cm}

\item
The ensemble exhibits a complete aggregation, if the configuration diameter tends to zero asymptotically:
\[ \lim_{t\rightarrow\infty} {\mathcal D}(T(t)) = 0, \]
where ${\mathcal D}(T) := \max_{i,j} \| T_i - T_j \|_F$ is the diameter of the configuration $\{T_j \}$. 

\vspace{0.2cm}

\item
The ensemble exhibits a practical aggregation, if the configuration diameter satisfies
\[  \lim_{\kappa_0 \to \infty+} \limsup_{t\rightarrow\infty} {\mathcal D}(T(t))=0.     \]
\end{enumerate}
\end{definition}

\subsection{The Lohe tensor model} \label{sec:2.1}
In this subsection, we briefly discuss the Lohe tensor model which generalizes earlier Lohe type models such as the Lohe sphere model on $\bbr^d$ and the Lohe matrix model on the unitary group $\bbu(d)$.  Although the detailed exposition on the modeling spirit can be found in \cite{H-P} in detail, we briefly provide the modeling spirit for the readability of the paper following aforementioned work. \newline

{{Consider an ensemble of particles in $\bbr^{d}$ which is an ensemble of real rank-1 tensors, and let $x_j = x_j(t)$ be the position of the $j$-th particle. Then, the Lohe sphere model on $\bbr^d$ reads as follows. 
\begin{equation} \label{NB-1}
\dot{x}_j=\Omega_j x_j +  \kappa (\langle{x_j, x_j}\rangle x_c-\langle{x_j, x_c}\rangle{}x_j), \quad j = 1, \cdots, N, 
\end{equation}
where $\Omega_j \in \bbr^{d \times d}$ is a real skew-symmetric matrix (real rank-2 tensor) with the property $\Omega_j^{t} = -\Omega_j$ and $\langle \cdot, \cdot \rangle$ is the standard inner product in $\bbr^d$.  These properties result in 
\begin{equation*} \label{NB-1-0}
\langle{x_j, \Omega_j{x_j}}\rangle+\langle{\Omega_jx_j, x_j}\rangle=0, \quad x_j \in \bbr^d.
\end{equation*}
System \eqref{NB-1} can be also rewritten in a component form:
\begin{equation} \label{NB-2-0-1}
\frac{d}{dt} [x_j]_{\alpha}= [\Omega_j]_{\alpha \beta} [x_j]_{\beta} + \kappa \Big ( [x_j]_\beta [x_j]_\beta  [x_c]_{\alpha} - [x_j]_\beta [x_c]_\beta  [x_j]_{\alpha}  \Big),
\end{equation}
where we used Einstein summation rule.  \newline

Next, we consider an ensemble of $d\times d$ unitary matrices. Then, the Lohe matrix model on $\bbu(d)$ reads as follows.
\begin{equation} \label{NB-2-1}
{\mathrm i}\dot{U}_jU_j^*=H_j+{{\mathrm i}\kappa\over2N}\sum_{k=1}^N\left( U_kU_j^*-U_jU_k^*\right),
\end{equation}
where  $A^*$ denotes the hermitian conjugate of the matrix $A$, and $H_j$ is the Hermitian matrix with the property $H_j^* = H_j$. These properties result in the following relation: 
\begin{equation*} \label{NB-2-2}
\langle U_j, -{\mathrm i} H_j U_j \rangle_F + \langle -{\mathrm i} H_j U_j, U_j \rangle_F =0, \quad j = 1, \cdots, N,
\end{equation*}
where $\langle \cdot, \cdot \rangle_F$ is the Frobenius inner product on ${\mathbb U}(d)$:
\[ \langle A, B \rangle_F := \mbox{tr}(A^* B), \quad A, B \in {\mathbb U}(d). \]
Then, it is well known from \cite{Lo-1, Lo-2} that system \eqref{NB-2-1} conserves the quadratic quantity $U_j^* U_j$. Hence one has 
\[ U_j^*(t) U_j(t)= I_d, \quad t \geq 0. \]
In this case, system \eqref{NB-2-1} becomes
\begin{equation} \label{NNB-2}
\frac{dU_j}{dt} = - {\mathrm i} H_j U_j +{\kappa\over2} \left( U_cU_j^*U_j -U_jU_c^* U_j\right).
\end{equation}
Then, the $(\alpha, \beta)$-th component of both sides of \eqref{NNB-2} satisfies
\begin{equation} \label{NB-2-2-4}
\frac{d}{dt} [U_j]_{\alpha \beta} = [-{\mathrm i} H_j U_j]_{\alpha \beta} + {\kappa\over2}\left[[U_c]_{\alpha\gamma }[U^*_j]_{\gamma \delta}[U_j]_{\delta\beta}-[U_j]_{\alpha\gamma} [U^*_c]_{\gamma \delta}[U_j]_{\delta\beta}\right].
\end{equation}
Motivated by the special cases \eqref{NB-2-0-1} and \eqref{NB-2-2-4}, the free flow part is given as 
\begin{equation} \label{NB-2-2-4-1}
[A_j]_{\alpha_{*0}\alpha_{*1}}[T_j]_{\alpha_{*1}}, 
\end{equation}
where $A_j$ is skew-Hermitian rank-$2m$ tensor such as
\[  \bar{[A_j]}_{\alpha_{*0}\alpha_{*1}}=-[A_j]_{\alpha_{*1}\alpha_{*0}}. \]
In contrast, motivated by the explicit forms \eqref{NB-2-0-1} and \eqref{NB-2-2-4}, interaction parts in \eqref{M-1} can be defined as cubic couplings :
\begin{equation} \label{NB-2-2-5}
\sum_{i_* \in \{0, 1\}^m}\kappa_{i_*} \Big([T_c]_{\alpha_{*i_*}}\bar{[T_j]}_{\alpha_{*1}}[T_j]_{\alpha_{*(1-i_*)}}-[T_j]_{\alpha_{*i_*}}\bar{[T_c]}_{\alpha_{*1}}[T_j]_{\alpha_{*(1-i_*)}} \Big).
\end{equation}
Finally, we combine \eqref{NB-2-2-4-1} and \eqref{NB-2-2-5} to derive the R.H.S. in \eqref{M-1}.  \newline

Next, we consider the coupling terms in \eqref{M-1}  for complex rank-1 and rank-2 tensors and compare them with those appearing in the Lohe sphere and Lohe matrix models. For complex rank-1 tensors, the terms in \eqref{NB-2-2-5} become
\begin{equation} \label{NB-2-2-6}
 \kappa_0 \Big(  [T_c]_{\alpha} [\bar{T}_i]_{\beta}[T_i]_{\beta} - [T_i]_{\alpha} [\bar{T}_c]_{\beta}[T_i]_{\beta}  \Big)  + \kappa_1 \Big( [T_c]_{\beta} [\bar{T}_j]_{\beta}[T_j]_{\alpha} - [T_j]_{\beta} [\bar{T}_c]_{\beta}[T_i]_{\alpha}  \Big ).
\end{equation}
Note that for real-valued tensors, the second term in \eqref{NB-2-2-6} becomes zero and the first term is exactly the same as in \eqref{NB-2-0-1}. Now we return to complex rank-2 tensors, the cubic interaction terms in  \eqref{NB-2-2-5} become
\begin{align*}
\begin{aligned} \label{NB-2-2-7}
&\kappa_{00}(\mathrm{tr}(T_j^* T_j)T_c-\mathrm{tr}(T_c^* T_j)T_j) + + \kappa_{11}\mathrm{tr}(T_j^* T_c-T_c^* T_j)T_j \\
&\hspace{2cm} + \kappa_{10}(T_j T_j^* T_c-T_j T_c^* T_j) +  \kappa_{01}(T_cT_j^* -T_jT_c^*) T_j.
\end{aligned}
\end{align*}}}
{Note that the terms $\kappa_{00}$ and $\kappa_{11}$ correspond to the coupling terms appearing in the Lohe sphere model, and the third term involving with $\kappa_{10}$ is exactly the same in the Lohe matrix model, and the forth term involving with $\kappa_{01}$ is a new interaction term.}

Although $2^m$ coupling terms in the R.H.S. of \eqref{M-1} look complicated, the system has a basic conservation law as follows.
\begin{lemma} \label{L2.1}
\emph{\cite{H-P}}
Let $\{ T_j \}$ be a global solution to the Lohe tensor model \eqref{M-1}. Then $\|T_j \|_F$ is a conserved quantity: for $j = 1, \cdots, N$, 
\[ \|T_j(t)\|_F = \|T_j^{in}\|_F, \quad t > 0. \]
\end{lemma}

\vspace{0.5cm}

\noindent Next, we summarize results in \cite{H-P} for system \eqref{M-1}. For this, we set 
\[  \|T_c^{in}\|_F := \|T_c(0)\|_F, \quad  \hat{\kappa}_0 :=  \sum_{i_*\neq0}\kappa_{i_*}, \]
where $T_c=\frac{1}{N}\sum_{k=1}^NT_k$ and $\eta$ is the largest root of the quadratic equation:
\[ 2\kappa_0 x^2+(\kappa_{0}- 4\hat{\kappa}_0  \|T_c^{in}\|_F^2 )x = {\mathcal D}(A). \]
Then, the root $\eta$ satisfies
\[  0< \eta < \frac{\kappa_0- 4\hat{\kappa}_0 \|T_c^{in}\|_F^2 }{2\kappa_0}.   \]
In next theorem, we present a sufficient framework for complete and practical aggregations.
\begin{theorem} \label{T2.1}
\emph{\cite{H-P}}
The following assertions hold.
\begin{enumerate}
\item
Suppose that frequency matrix, the coupling strengths and initial data satisfy
\begin{align*}
\begin{aligned} 
& A_j = 0,  \quad  {\hat \kappa}_{0}  < \frac{\kappa_{0}}{4 \|T_c^{in}\|_F^2}, \quad \|T_j^{in}\|_F = 1, \quad j = 1, \cdots, N, \\
& 0< {\mathcal D}(T^{in})<\frac{\kappa_{0}- 4{\hat \kappa}_0 \|T_c^{in}\|_F^2}{2 \kappa_0}, 
\end{aligned}
\end{align*}
and let $\{ T_i \}$ be a global solution to \eqref{M-1}. Then, there exist positive constants $C_0$ and $C_1$ depending on $\kappa_{i_*}$ and $\{ T_j ^{in} \}$ such that 
\begin{equation*} \label{F-1-0}
C_0 e^{-\left(\kappa_{0}+ 4 \hat{\kappa}_0  \|T_c^{in}\|_F^2 \right)t} \leq {\mathcal D}(T(t)) \leq C_1 e^{-\left(\kappa_{0}- 4 \hat{\kappa}_0 \|T_c^{in}\|_F^2  \right)t}, \quad t \geq 0.
\end{equation*}
\item
Suppose that coupling strength, initial data and frequency matrices satisfy
\[ \kappa_0 > 0, \quad  0\leq{} {\mathcal D}(T(0))\leq\eta \quad \mbox{and} \quad {\mathcal D}(A)< \frac{|\kappa_0- 4\hat{\kappa}_0 \|T_c^{in}\|_F^2 |^2}{8 \kappa_0}, \]
and let $\{T_i \}$ be a solution to system \eqref{D-1}. Then practical synchronization emerges asymptotically:
\[ \lim_{D(A)/\kappa_0 \to 0+} \limsup_{t\rightarrow\infty}D(T(t))=0.  \]
\end{enumerate}
\end{theorem}
\begin{proof}
(i)~The first assertion is based on Gronwall's differential inequality:
\[   \left|{d\over{dt}} {\mathcal D}(T)+\kappa_0 {\mathcal D}(T) \right|  \leq 2\kappa_0{\mathcal D}(T)^2+ 2\hat{\kappa}_0 \|T_c^{in}\|_F {\mathcal D}(T), \quad \mbox{a.e.}~t > 0. \]
(ii)~The second assertion is also based on Gronwall's differential inequality:
\[
\frac{d}{dt} {\mathcal D}(T) \leq 2\kappa_0 {\mathcal D}(T)^2-(\kappa_{0}- 4\hat{\kappa}_0  \|T_c^{in}\|^2_F) {\mathcal D}(T)+ {\mathcal D}(A), \quad \mbox{a.e.~$t > 0$}.
\]
\end{proof}
{{
\begin{remark} \label{R2.1}
Note that if $\|T_i^{in}\|_2\neq\|T_j^{in}\|_F$ for some $i\neq j$, then by Lemma \ref{L2.1} one has 
 \[ \|T_i(t)\|_F\neq \|T_j(t)\|_F \quad \mbox{for all $t$}. \]
 Thus, the relation 
\[ \lim_{t \to \infty} \| T_i(t) -T_j(t) \|_F \not = 0,
\]
cannot be achived, i.e., the ensemble cannot aggregate. This is why we impose on the unit norm condition $ \|T_j^{in}\|_F = 1$. 
\end{remark}}}

\subsection{The Lohe hermitian sphere model}  \label{sec:2.2} 
Consider the Lohe hermitian sphere model:
\begin{equation} \label{A-4}
\dot{z}_j= \Omega_j z_j +\kappa_{0} (z_c\langle z_j, z_j \rangle-z_j \langle z_c, z_j \rangle )+\kappa_1(\langle{z_j, z_c}\rangle- \langle z_c,  z_j\rangle) z_j.
\end{equation}
Note that the interaction terms involving with $\kappa_0$ and $\kappa_1$ in \eqref{A-4} correspond to the Lohe sphere model and new terms due to the complex nature of underlying state space. For example, for a real rank-1 tensor $z_j = x_j \in \bbr^d$, the second coupling term in the R.H.S. of \eqref{A-4} becomes zero:
\begin{equation*} \label{A-5}
\langle{z_j, z_c}\rangle- \langle z_c, z_j\rangle = 0.
\end{equation*}
Hence, for $z_j = x_j  \in \bbr^d$, system \eqref{A-4} reduces to the Lohe sphere model \cite{Lo-2}:
\[ 
\dot{x}_j= \Omega_j x_j +\kappa_{0} \Big(x_c\langle x_j, x_j \rangle-x_j \langle x_c. x_j \rangle \Big).
\]
Now, we return to the Lohe hermitian sphere model \eqref{A-4}, and explain how \eqref{A-4} can be related to the Kuramoto model. For this, we set 
\begin{equation} \label{A-6}
\Omega_j = 0, \quad  z_j = r_j e^{\mathrm{i}\theta_j}, \quad j = 1, \cdots, N \quad \mbox{and} \quad z_c := r_c e^{\mathrm{i}\phi}.
 \end{equation}
We substitute the ansatz \eqref{A-6} into \eqref{A-4} to see
\begin{align*}
\begin{aligned}
\dot{r}_je^{\mathrm{i}\theta_j}+{\mathrm i} r_je^{\mathrm{i}\theta_j}\dot{\theta_j}&=\kappa_0 r_j^2 r_c (e^{{\mathrm i} \phi}-e^{{\mathrm i} (2\theta_j-\phi)})+\kappa_1 r_j^2 r_c (e^{\mathrm i \phi}-e^{\mathrm i (2\theta_j-\phi)})\\
&=2(\kappa_0+\kappa_1) r_j^2 r_c \mathrm{i}\sin(\phi-\theta_j)e^{\mathrm i \theta_j}.
\end{aligned}
\end{align*}
This yields
\begin{equation} \label{A-7}
\dot{r}_j+ {\mathrm i} r_j\dot{\theta_j} =2(\kappa_0+\kappa_1)\mathrm{i} r_j^2 r_c \sin(\phi-\theta_j).
\end{equation}
We compare real and imaginary parts of the above relation \eqref{A-7} to get 
\[
\dot{r}_j=0 \quad \mbox{and} \quad \dot{\theta}_j=2(\kappa_0+\kappa_1)r_j r_c \sin(\phi-\theta_j).
\]
Hence, we have
\[ r_j(t) = r_j^{in} \qquad  \dot{\theta_j}=\frac{2(\kappa_0+\kappa_1)}{N}\sum_{k=1}^N r_j^{in} r_k^{in}\sin(\theta_k-\theta_j). \]
Next, we discuss the solution splitting property of \eqref{A-4} for a homogeneous ensemble with $\Omega_j = \Omega$.  Consider two Cauchy problems with the same initial data:
\begin{equation} \label{B-1}
\begin{cases}
\displaystyle \dot{z}_j = \Omega z_j +  \kappa_{0} \Big(z_c\langle z_j, z_j \rangle-z_j \langle z_c. z_j \rangle \Big)+\kappa_1 \Big(\langle{z_j, z_c}\rangle- \langle z_c, z_j\rangle \Big) z_j, \quad t > 0, \\
\displaystyle z_j(0)  = z_j^{in}, \quad j = 1, \cdots, N,
\end{cases}
\end{equation}
and
\begin{equation} \label{B-1-0-0}
\begin{cases}
\dot{w}_j=  \kappa_{0} \Big(w_c\langle w_j, w_j \rangle-w_j \langle w_c. w_j \rangle \Big)+\kappa_1 \Big(\langle{w_j, w_c}\rangle- \langle w_c, w_j\rangle \Big) w_j, \quad t > 0, \\
\displaystyle w_j(0)  = z_j^{in}, \quad j = 1, \cdots, N.
\end{cases}
\end{equation}
In next proposition, we show that  how solutions to \eqref{B-1} and \eqref{B-1-0-0} are related each other.
\begin{proposition} \label{P2.1}
Let $\{z_j \}$ and $\{w_j \}$ be solutions to \eqref{B-1} and \eqref{B-1-0-0} with the same initial data $\{z_j^{in} \}$. Then, one has 
\[ z_j=e^{\Omega t} w_j, \quad j = 1, \cdots, N.    \]
\end{proposition}
\begin{proof}
Note that 
\[  \left( e^{\Omega t} \right)^* = (e^{\Omega t})^{-1}.
\]
Then $e^{\Omega t}$ is unitary, and we introduce the variable $y_j$ such that 
\begin{equation} \label{B-1-1}
 z_j=e^{\Omega t} w_j \quad \mbox{for all $j=1, 2, \cdots, N$}.
\end{equation} 
We substitute \eqref{B-1-1} into system \eqref{B-1} to get 
\begin{align*}
\begin{aligned}
e^{\Omega t}\dot{w}_j+\Omega e^{\Omega t}w_j &=\Omega {e^{\Omega t}}w_j+\kappa_0(\langle{{e^{\Omega t}}w^j, {e^{\Omega t}}w_j}\rangle{e^{\Omega t}}w_c-\langle{{e^{\Omega t}}w_c, {e^{\Omega t}}w_j}\rangle {e^{\Omega t}}w_j) \\
&\hspace{1cm} +\kappa_1(\langle{{e^{\Omega t}}w_j, {e^{\Omega t}}w_c}\rangle-\langle{{e^{\Omega t}}w_c, {e^{\Omega t}}w_c}\rangle){e^{\Omega t}}w_j.
\end{aligned}
\end{align*}
After simplification, one has
\begin{equation*} \label{B-1-2}
\dot{w}_j=\kappa_0(\langle{w_j, w_j}\rangle w_c-\langle{w_c,w_j}\rangle w_j)+\kappa_1(\langle{w_j, w_c}\rangle-\langle{w_c, w_j}\rangle)w_j.
\end{equation*}
Thus, we obtain the desired result.
\end{proof}

\section{Emergent dynamics of Subsystem A and Subsystem B} \label{sec:3}
\setcounter{equation}{0}
In this section, we study emergent behaviors of two sub-systems to the complex Lohe sphere model generalizing the Lohe sphere model on the $\bbc^d$. In the following two subsections, we study the emergent dynamics of the following two subsystems:
\begin{align}
\begin{aligned} \label{C-0}
\mbox{Subsystem A}: \quad & \dot{z}_j = \kappa_0(\langle{z_j, z_j}\rangle z_c-\langle{z_c, z_j}\rangle z_j), \quad t > 0, \\
 &z_j(0) =z_j^{in},\quad {{|z_j^{in} |=1}}, \quad  j=1, 2 , \cdots, N, \\
\mbox{Subsystem B}: \quad &\dot{z}_j =\kappa_1(\langle{z_j, z_c}\rangle-\langle{z_c, z_j}\rangle)z_j, \quad t > 0, \\
&z_j(0) =z_j^{in},\quad {{|z_j^{in} |=1}}, \quad \quad j=1, 2 , \cdots, N.
\end{aligned}
\end{align}
\subsection{Subsystem A} \label{sec:3.1} In this subsection, we discuss the following two items for the Subsystem A:
\begin{itemize}
\item
Existence of the cross-ratio like conserved quantities.
\vspace{0.1cm}.
\item
Emergence of complete aggregation.
\end{itemize}

\vspace{0.2cm}

\subsubsection{Constant of motion}~Let $\{z_j \}$ be a solution to the $\eqref{C-0}_1$ with
\begin{equation} \label{B-3-0}
 z_i \not = z_j, \quad {{|z_i| = 1}}, \quad  1 \leq i < j \leq N.
\end{equation} 
Then, the two-point correlation function $\langle{z_i, z_j}\rangle$ satisfies
\begin{align*}
\frac{d}{dt}\langle{z_i, z_j}\rangle&=\langle{z_i,\dot{z}_j}\rangle+\langle{\dot{z}_i, z_j}\rangle\\
&=\kappa_0\left(\langle{z_j, z_j}\rangle\langle{z_i, z_c}\rangle-\langle{z_c, z_j}\rangle\langle{z_i, z_j}\rangle\right)+\kappa_0(\langle{z_i, z_i}\rangle\langle{z_c, z_j}\rangle-\langle{z_i, z_c}\rangle\langle{z_i, z_j}\rangle)\\
&=\kappa_0(\langle{z_i, z_c}\rangle+\langle{z_c, z_j}\rangle-\langle{z_c, z_j}\rangle\langle{z_i, z_j}\rangle-\langle{z_i, z_c}\rangle\langle{z_i, z_j}\rangle)\\
&=\kappa_0(1-\langle{z_i, z_j}\rangle)(\langle{z_i, z_c}\rangle+\langle{z_c, z_j}\rangle).
\end{align*}
This yields
\begin{align}\label{B-3}
\frac{d}{dt}(1-\langle{z_i, z_j}\rangle)=-\kappa_0(1-\langle{z_i, z_j}\rangle)(\langle{z_i, z_c}\rangle+\langle{z_c, z_j}\rangle).
\end{align}
For four-tuple of indices $i,j, k, l$, we define a functional $\mathcal{C}_{ijkl}$:
\begin{equation} \label{B-3-2}
\mathcal{C}_{ijkl} :=\frac{(1-\langle{z_i, z_j}\rangle)(1-\langle{z_k, z_l}\rangle)}{(1-\langle{z_i, z_l}\rangle)(1-\langle{z_k, z_j}\rangle)}.
\end{equation}
\begin{proposition} \label{P3.1}
Let $\{ z_j \}$ be a global solution to $\eqref{C-0}_1$ with the non-overlapping property \eqref{B-3-0}. Then, the functional $ \mathcal{C}_{ijkl}$ is a constant of motion:
\[  \mathcal{C}_{ijkl}(t) = \mathcal{C}_{ijkl}(0), \quad t \geq 0. \]
\end{proposition}
\begin{proof}
It follows from \eqref{B-3} that 
\begin{equation} \label{B-3-2}
\frac{\frac{d}{dt}(1-\langle{z_i, z_j}\rangle)}{1-\langle{z_i, z_j}\rangle}=-\kappa_0(\langle{z_i, z_c}\rangle+\langle{z_c, z_j}\rangle).
\end{equation}
Then we use \eqref{B-3-2} to get 
\[ \frac{d}{dt}(1-\langle{z_i, z_j}\rangle)=-\kappa_0(1-\langle{z_i, z_j}\rangle)(\langle{z_i, z_c}\rangle+\langle{z_c, z_j}\rangle), \]
or equivalently,
\[ \frac{d}{dt}\frac{1}{(1-\langle{z_i, z_j}\rangle)}=\frac{\kappa_0}{(1-\langle{z_i, z_j}\rangle)}(\langle{z_i, z_c}\rangle+\langle{z_c, z_j}\rangle).
\]
These estimates yield  
\begin{align*}
&\frac{d}{dt}\frac{(1-\langle{z_i, z_j}\rangle)(1-\langle{z_k, z_l}\rangle)}{(1-\langle{z_i, z_l}\rangle)(1-\langle{z_k, z_j}\rangle)} =\frac{1-\langle{z_k, z_l}\rangle}{(1-\langle{z_i, z_l}\rangle)(1-\langle{z_k, z_j}\rangle)}\frac{d}{dt}(1-\langle{z_i, z_j}\rangle) \\
& \hspace{1.2cm} +\frac{1-\langle{z_i, z_j}\rangle}{(1-\langle{z_i, z_l}\rangle)(1-\langle{z_k, z_j}\rangle)}\frac{d}{dt}(1-\langle{z_k, z_l}\rangle)\\
& \hspace{1.2cm} +\frac{(1-\langle{z_i, z_j}\rangle)(1-\langle{z_k, z_l}\rangle)}{1-\langle{z_k, z_j}\rangle}\frac{d}{dt}\left(\frac{1}{1-\langle{z_i, z_l}\rangle}\right) \\
& \hspace{1.2cm} +\frac{(1-\langle{z_i, z_j}\rangle)(1-\langle{z_k, z_l}\rangle)}{1-\langle{z_i, z_l}\rangle}\frac{d}{dt}\left(\frac{1}{1-\langle{z_k, z_j}\rangle}\right)\\
& \hspace{1cm} =\frac{(1-\langle{z_i, z_j}\rangle)(1-\langle{z_k, z_l}\rangle)}{(1-\langle{z_i, z_l}\rangle)(1-\langle{z_k, z_j}\rangle)} \\
& \hspace{1cm} \times \left(\frac{\frac{d}{dt}(1-\langle{z_i, z_j}\rangle)}{1-\langle{z_i, z_j}\rangle}+\frac{\frac{d}{dt}(1-\langle{z_k, z_l}\rangle)}{1-\langle{z_k, z_l}\rangle}-\frac{\frac{d}{dt}(1-\langle{z_i, z_l}\rangle)}{1-\langle{z_i, z_l}\rangle}-\frac{\frac{d}{dt}(1-\langle{z_k, z_j}\rangle)}{1-\langle{z_k, z_j}\rangle}\right)=0.
\end{align*}
\end{proof}

\begin{remark}
\noindent For a configuration $Z= \{z_j\}$ with the non-overlapping condition \eqref{B-3-0}, we define another cross ratio-like functional:
\[
R_{ijkl} := |{\mathcal C}_{ijkl} | =  \frac{|1-\langle{z_i, z_j}\rangle|\cdot|1-\langle{z_k, z_l}\rangle|}{|1-\langle{z_j, z_k}\rangle|\cdot|1-\langle{z_l, z_i}\rangle|} .
\]
As a corollary of Proposition \ref{P3.1}, one has the conservation of $R_{ijkl}$ as well:
\[ R_{ijkl}(t) =  |{\mathcal C}_{ijkl}(t) | =  |{\mathcal C}_{ijkl}(0)| = R_{ijkl}(0), \quad t \geq 0. \] 
\end{remark}

\vspace{0.5cm}

\subsubsection{Emergence of complete aggregation}~In this subsection, we study the emergent dynamics of $\eqref{C-0}_1$. 

First note that 
\begin{equation} \label{B-4}
 \|z_i - z_j \|^2 = \| z_i \|^2 + \|z_j \|^2 - (\langle z_i, z_j \rangle + \overline{\langle z_i, z_j \rangle} ) 
= 2\mathrm{Re}(1-\langle{z_i, z_j}\rangle)\leq2|1-\langle{z_i,z_j}\rangle|. 
\end{equation}
This implies
\begin{align}\label{B-5}
2\ln\|z_i-z_j\|\leq\ln2+\ln|1-\langle{z_i, z_j}\rangle|.
\end{align}
Based on \eqref{B-4}, we introduce a functional:
\[
\mathcal{D}(Z) :=\max_{1\leq i, j \leq N}|1-\langle{z_i, z_j}\rangle|.
\]
It follows from \eqref{B-3}  that 
\begin{align*}
\frac{d}{dt}[(1-\langle{z_i, z_j}\rangle)(1-\langle{z_j, z_i}\rangle)]=-\kappa_0(1-\langle{z_i, z_j}\rangle)(1-\langle{z_j, z_i}\rangle)(\langle{z_c, z_i+z_j}\rangle+\langle{z_i+z_j, z_c}\rangle).
\end{align*}
Since $1-\langle{z_i, z_j}\rangle$ is the complex conjugate of $1-\langle{z_j, z_i}\rangle$, we have
\[
(1-\langle{z_i, z_j}\rangle)(1-\langle{z_j, z_i}\rangle)=|1-\langle{z_i, z_j}\rangle|^2.
\]
If $\langle{z_i, z_j}\rangle\neq1$, we have
\begin{align}\label{B-5-1}
\frac{d}{dt}\ln|1-\langle{z_i, z_j}\rangle|=-\frac{\kappa_0}{2}(\langle{z_c, z_i+z_j}\rangle+\langle{z_i+z_j, z_c}\rangle).
\end{align}
\begin{lemma} \label{L3.1}
Let $\{ z_j \}$ be a solution to $\eqref{C-0}_1$. Then, the following assertions hold.
\begin{enumerate}
\item
The order parameter $\rho$ defined in \eqref{M-9} is non-decreasing:
\[ 
\frac{d\rho}{dt} \geq0, \quad \forall~t > 0.
\]
\item
The functional $\sum_{1\leq i, j \leq N}\ln|1-\langle{z_i, z_j}\rangle|$ is non-increasing:
\[ \frac{d}{dt} \sum_{1\leq i, j \leq N}\ln|1-\langle{z_i, z_j}\rangle|=-2\kappa_0 N^2 \rho^2\leq0, \quad \forall~t > 0. \]
\end{enumerate}
\end{lemma}
\begin{proof}
\noindent (i)~It follows from \eqref{B-3} that
\begin{align*}
\frac{d}{dt}(1-\langle{z_c, z_c}\rangle)&=-\frac{\kappa_0}{N^2}\sum_{i, j=1}^N(1-\langle{z_i, z_j}\rangle)(\langle{z_i, z_c}\rangle+\langle{z_c, z_j}\rangle)\\
&=-\frac{\kappa_0}{N^2}\sum_{i, j=1}^N(\langle{z_i, z_c}\rangle+\langle{z_c, z_j}\rangle-\langle{z_i, z_j}\rangle\langle{z_i, z_c}\rangle-\langle{z_i, z_j}\rangle\langle{z_c, z_j}\rangle)\\
&=-\kappa_0(\langle{z_c, z_c}\rangle+\langle{z_c, z_c}\rangle)+\frac{\kappa_0}{N}\left(\sum_{i=1}^N\langle{z_i, z_c}\rangle^2+\sum_{j=1}^N\langle{z_c, z_j}\rangle^2\right)\\
&=-2\kappa_0\|z_c\|^2+\frac{2\kappa_0}{N}\sum_{i=1}^N\mathrm{Re}\left(\langle{z_i, z_c}\rangle^2\right)\\
&=-\frac{2\kappa_0}{N}\sum_{i=1}^N\left(\|z_c\|^2-\mathrm{Re}\left(\langle{z_i, z_c}\rangle^2\right)\right).
\end{align*}
On the other hand, by Cauchy-Schwarz inequality, one has
\[
|\langle{z_i, z_c}\rangle|^2\leq\langle{z_i, z_i}\rangle\langle{z_c, z_c}\rangle.
\]
Hence we have 
\[
\mathrm{Re}\left(\langle{z_i, z_c}\rangle^2\right)\leq|\langle{z_i, z_c}\rangle|^2\leq \langle{z_i, z_i}\rangle\langle{z_c, z_c}\rangle=\|z_c\|^2.
\]
This yields the desired estimate:
\[
\frac{d}{dt}(1-\langle{z_c, z_c}\rangle)\leq0 \quad \mbox{or equivalently} \quad \frac{d \rho^2}{dt} \geq 0.
\]
Thus, $\rho$ is non-decreasing. \newline

\noindent (ii)~We use \eqref{B-5-1} to get
\begin{align*}
\begin{aligned}
\frac{d}{dt}\frac{1}{N^2}\sum_{1\leq i, j \leq N}\ln|1-\langle{z_i, z_j}\rangle| &=-\frac{\kappa_0}{2}(\langle{z_c, z_c+z_c}\rangle+\langle z_c, z_c+z_c\rangle) \\
&=-2\kappa_0\langle z_c, z_c\rangle=-2\kappa_0\rho^2 \leq 0.
\end{aligned}
\end{align*}
Thus, we have the desired estimate.
\end{proof}
Now, we are ready to state the result on the emergent dynamics of Subsystem A $\eqref{C-0}_1$.
\begin{theorem} \label{T3.1}
Suppose that the coupling strength and initial data satisfy
\[
\kappa_0 > 0, \quad \|z_i^{in} \|=1,\quad \lambda_M(0) :=\max_{i\neq j}|1-\langle{z_i^{in}, z_j^{in}}\rangle|<1/2.
\]
Then, for a global solution $\{ z_j \}$ to $\eqref{C-0}_1$, we have an exponential aggregation: there exists a positive constant $\Lambda$ depending on initial data such that 
\[ {\mathcal D}(Z(t)) \leq {\mathcal D}(Z^{in}) e^{- \kappa_0 \Lambda t}, \quad t \geq 0. \]
\end{theorem}
\begin{proof} Note that it suffices to derive the following estimate: for $i, j = 1, \cdots, N$,
\begin{equation} \label{B-5-2}
|1-\langle{z_i(t), z_j(t)}\rangle|\leq|1-\langle{z_i^{in}, z_j^{in}}\rangle|e^{-\kappa_0(1-2\lambda_M(0))t}.
\end{equation}
For a given configuration $Z = \{z_j\}$, we set 
\[ h_{ij} = \langle z_i, z_j \rangle, \quad R_{ij}:= \mathrm{Re}\langle{z_i, z_j}\rangle, \quad I_{ij}:= \mathrm{Im}\langle{z_i, z_j}\rangle. \]
\noindent $\bullet$~Step A (Derivation of dynamics for $R_{ij}$ and $I_{ij}$): Note that $h_{ij}$ satisfies  
\begin{align}
\begin{aligned} \label{B-5}
\frac{d}{dt}\langle{z_i, z_j}\rangle&=\kappa_0(1-\langle{z_i, z_j}\rangle)(\langle{z_i, z_c}\rangle+\langle{z_c, z_j}\rangle)\\
&=\kappa_0(\langle{z_i, z_c}\rangle+\langle{z_c, z_j}\rangle-\langle{z_i, z_j}\rangle\langle{z_i, z_c}\rangle-\langle{z_i, z_j}\rangle\langle{z_c, z_j}\rangle).
\end{aligned}
\end{align}
We compare the real and imaginary parts of system \eqref{B-5} to get 
\begin{align*}
\frac{d}{dt}\mathrm{Re}\langle{z_i, z_j}\rangle&=\kappa_0(\mathrm{Re}\langle{z_i, z_c}\rangle+\mathrm{Re}\langle{z_c, z_j}\rangle-\mathrm{Re}\langle{z_i, z_j}\rangle\mathrm{Re}\langle{z_i, z_c}\rangle-\mathrm{Re}\langle{z_i, z_j}\rangle\mathrm{Re}\langle{z_c, z_j}\rangle\\
& \hspace{0.2cm} +\mathrm{Im}\langle{z_i, z_j}\rangle\mathrm{Im}\langle{z_i, z_c}\rangle+\mathrm{Im}\langle{z_i, z_j}\rangle\mathrm{Im}\langle{z_c, z_j}\rangle),
\end{align*}
\begin{align*}
\frac{d}{dt}\mathrm{Im}\langle{z_i, z_j}\rangle&=\kappa_0(\mathrm{Im}\langle{z_i, z_c}\rangle+\mathrm{Im}\langle{z_c, z_j}\rangle-\mathrm{Im}\langle{z_i, z_j}\rangle\mathrm{Re}\langle{z_i, z_c}\rangle-\mathrm{Im}\langle{z_i, z_j}\rangle\mathrm{Re}\langle{z_c, z_j}\rangle\\
& \hspace{0.2cm}-\mathrm{Re}\langle{z_i, z_j}\rangle\mathrm{Im}\langle{z_i, z_c}\rangle-\mathrm{Re}\langle{z_i, z_j}\rangle\mathrm{Im}\langle{z_c, z_j}\rangle),
\end{align*}
or equivalently
\begin{align*}
\begin{aligned}
\frac{d}{dt}R_{ij}&=\kappa_0(R_{ic}+R_{cj}-R_{ij}R_{ic}-R_{ij}R_{cj}+I_{ij}I_{ic}+I_{ij}I_{cj})\\
&=\kappa_0(1-R_{ij})(R_{ic}+R_{cj})+\kappa_0I_{ij}(I_{ic}+I_{cj}), \\
\frac{d}{dt}I_{ij}&=\kappa_0(I_{ic}+I_{cj}-I_{ij}R_{ic}-I_{ij}R_{cj}-R_{ij}I_{ic}-R_{ij}I_{cj})\\
&=\kappa_0(1-R_{ij})(I_{ic}+I_{cj})-\kappa_0I_{ij}(R_{ic}+R_{cj}).
\end{aligned}
\end{align*}
Since we expect $R_{ij} \to 1$, we introduce the quantity:
\[  J_{ij} :=1-R_{ij}. \]
Then, one has the following systems for $(J_{ij}, I_{ij})$:
\begin{align*}
\begin{cases}
\dot{J}_{ij}=-\kappa_0(2-J_{ic}-J_{cj})J_{ij}+\kappa_0(I_{ic}+I_{cj})I_{ij},\\
\dot{I}_{ij}=\kappa_0(I_{ic}+I_{cj})J_{ij}-\kappa_0(2-J_{ic}-J_{cj})I_{ij},
\end{cases}
\end{align*}
or we can write vector form:
\begin{align*}
\frac{d}{dt}
\begin{bmatrix}
J_{ij}\\
I_{ij}
\end{bmatrix}
=
\begin{bmatrix}
-\kappa_0(2-J_{ic}-J_{cj})&\kappa_0(I_{ic}+I_{cj})\\
\kappa_0(I_{ic}+I_{cj})&-\kappa_0(2-J_{ic}-J_{cj})
\end{bmatrix}
\begin{bmatrix}
J_{ij}\\
I_{ij}
\end{bmatrix}.
\end{align*}
We set $\alpha_{ij}$ and $\beta_{ij}$:
\[
\alpha_{ij}=\kappa_0(2-J_{ic}-J_{cj}), \qquad \beta_{ij}=\kappa_0(I_{ic}+I_{cj}).
\]
Then, we can obtain
\begin{equation} \label{B-6}
\frac{d}{dt}
\begin{bmatrix}
J_{ij}\\
I_{ij}
\end{bmatrix}
=
\begin{bmatrix}
-\alpha_{ij}&\beta_{ij}\\
\beta_{ij}&-\alpha_{ij}
\end{bmatrix}
\begin{bmatrix}
J_{ij}\\
I_{ij}
\end{bmatrix}.
\end{equation}

\vspace{0.5cm}

\noindent $\bullet$~Step B (Decay estimates for $I_{ij}$ and $J_{ij}$): Next, we will derive
\begin{equation} \label{I-1}
\frac{d}{dt}(I_{ij}^2+J_{ij}^2)\leq-2\kappa_0(1-2\lambda)(I_{ij}^2+J_{ij}^2). 
\end{equation}
{\it Proof of \eqref{I-1})}: We use \eqref{B-6} to have
\begin{align}
\begin{aligned} \label{B-7}
\frac{d}{dt}(I_{ij}^2+J_{ij}^2)&=-\alpha_{ij}I_{ij}^2+2\beta_{ij}I_{ij}J_{ij}-\alpha_{ij}J_{ij}^2 \\
&=-(\alpha_{ij}-\beta_{ij})(I_{ij}^2+J_{ij}^2)-\beta_{ij}(I_{ij}-J_{ij})^2\\
&\leq-(\alpha_{ij}-\beta_{ij})(I_{ij}^2+J_{ij}^2).
\end{aligned}
\end{align}
We set
\[ \lambda(t)=\max_{i\neq j}\sqrt{I_{ij}(t)^2+J_{ij}(t)^2}. \]
Suppose that 
\[ 0\leq\lambda(0)<\frac{1}{2}. \]
Since $|I_{ij}|, |J_{ij}|\leq\lambda$, we can see
\begin{align}
\begin{aligned} \label{B-8}
\alpha_{ij}-\beta_{ij} &=\kappa_0(2-I_{ic}-I_{cj}-J_{ic}-J_{cj}) \geq\kappa_0(2-\lambda-\lambda-\lambda-\lambda) \\
&=2\kappa_0(1-2\lambda).
\end{aligned}
\end{align}
We combine \eqref{B-7} and \eqref{B-8} to get 
\begin{align}\label{I-1}
\frac{d}{dt}(I_{ij}^2+J_{ij}^2)\leq-2\kappa_0(1-2\lambda)(I_{ij}^2+J_{ij}^2).
\end{align}
We also set $i_M$ and $j_M$ satisfy:
\[ \lambda(t) :=\sqrt{I_{i_Mj_M}(t)^2+J_{i_Mj_M}(t)^2}, \quad t \in [0, T). \]
Then we can obtain following inequality:
\begin{align*}
\begin{aligned}
\frac{d}{dt}\lambda^2 &=\frac{d}{dt}(I_{i_Mj_M}^2+J_{i_Mj_M}^2)\leq-2\kappa_0(1-2\lambda)(I_{i_Mj_M}^2+J_{i_Mj_M}^2) \\
&=-2\kappa_0(1-2\lambda)\lambda^2, \quad t\in[0, T),
\end{aligned}
\end{align*}
which yields
\[
\frac{d}{dt}\lambda_M \leq-\kappa_0(1-2\lambda)\lambda, \quad t\in[0, T).
\]
After simplification, one has 
\[
\lambda(t)\leq\frac{\lambda(0)}{2\left(\left(\frac{1}{2}-\lambda(0)\right)e^{\kappa_0t}+\lambda(0)\right)}, \quad t\in[0, T),
\]
Therefore we have
\begin{equation} \label{B-9}
\lambda_M(t) \leq \lambda(0), \quad t\in[0, T).
\end{equation}
We combine \eqref{B-7} and \eqref{B-9} to obtain
\[
\frac{d}{dt}(I_{ij}^2+J_{ij}^2)\leq-2\kappa_0(1-2\lambda(0))(I_{ij}^2+J_{ij}^2).
\]
This yields
\[
I_{ij}(t)^2+J_{ij}(t)^2\leq(I_{ij}(0)^2+J_{ij}(0)^2)e^{-2\kappa_0(1-2\lambda_M(0))t},
\]
or equivalently, 
\[
|1-\langle{z_i(t), z_j(t)}\rangle|^2\leq|1-\langle{z_i(0), z_j(0)}\rangle|^2e^{-2\kappa_0(1-2\lambda_M(0))t}.
\]
Thus, one has the desired estimate.
\end{proof}
\begin{remark} \label{R3.2}
1. As a direct corollary and \eqref{B-5-2}, one has 
\[
\|z_i-z_j\|^2\leq2|1-\langle{z_i(0), z_j(0)}\rangle|e^{-\kappa_0(1-2\lambda_M(0))t}.
\]

\noindent {2. For real case $z_j \in \bbr^d$, note that the condition 
\[ \lambda_M(0) :=\max_{i\neq j}|1-\langle{z_i^{in}, z_j^{in}}\rangle|<1/2 \]
is equivalent to 
\[ {\mathcal D}(Z^{in}) < 1.      \]
This is certainly weaker than that of \cite{C-C-H} in which the exponential aggregation is employed.
\[   {\mathcal D}(Z^{in}) < \frac{1}{4}. \]
In fact, gradient flow formulation of Subsystem A yields that state-locking emerges from any generic initial data without any convergence rate. In this sense, our result is weaker than that of \cite{H-K-R}.}
\end{remark}
\subsection{Subsystem B} \label{sec:3.2}
 In this subsection, we discuss the following two items for Subsystem B: \newline
\begin{itemize}
\item
Equivalence of Subsystem B and the Kuramoto type model with frustration.
\vspace{0.2cm}
\item
A gradient flow formulation of the Kuramoto type model with frustration.
\end{itemize}

\vspace{0.2cm}

\subsubsection{The Kuramoto dynamics with frustration} \label{sec:3.21} In this part, we show that how system $\eqref{C-0}_2$ can be transformed into a Kuramoto type model with skew-symmetric frustration. These results can be summarized as follows. 
\begin{theorem} \label{T3.2}
Let $\{ z_j \}$ be a global solution to $\eqref{C-0}_2$ with the initial data $\{z_j^{in} \}$. Then, the following assertions hold.
\begin{enumerate}
\item
There exists a {real} time-dependent phase $\theta_j$ such that 
\begin{equation} \label{B-11}
 z_j(t)= e^{{\mathrm i} \theta_j(t)} z^{in}_j, \quad j = 1, \cdots, N. 
 \end{equation}
\item
If we set {real} $R_{jk}^{in}$ and $\alpha_{ji}$ such that
\begin{equation} \label{B-11-0}
 \langle{z_j^{in}, z_k^{in}} \rangle=: R^{in}_{jk}e^{\mathrm{i} \alpha_{jk}}, 
\end{equation} 
then the phase $\theta_j$ in (1) is a solution to the following Cauchy problem:
\begin{equation} \label{B-11-1}  
\begin{cases}
\displaystyle \dot{\theta}_j=\frac{2 \kappa_1}{N}\sum_{k=1}^N R_{jk}^{in} \sin(\theta_k-\theta_j+\alpha_{jk}), \quad t > 0, \\
\displaystyle \theta_j(0)=0,
\end{cases}
\end{equation}
where $R_{jk}^{in}$ and $\alpha_{jk}$ satisfy symmetry and anti-symmetry properties:
\[   R_{jk}^{in} = R_{kj}^{in}, \qquad \alpha_{jk} = -\alpha_{kj}, \quad \forall~~k, j = 1, \cdots, N. \]
\end{enumerate}
\end{theorem}
\begin{proof}
\noindent (i)~It follows from $\eqref{C-0}_2$ that 
\[
\dot{x}_j=2\kappa_1\mathrm{i}\mathrm{Im}(\langle{z_j, z_c}\rangle)z_j.
\]
In component wise, the above relation can be rewritten as
\[
\frac{d}{dt} [z_j]_\alpha=2\kappa_1\mathrm{i}\mathrm{Im}(\langle{z_j, z_c}\rangle)[z_j]_\alpha.
\]
This implies
\begin{align}
\begin{aligned} \label{B-11-2}
& [z_j(t)]_\alpha=[z_j(0)]_\alpha e^{2\kappa_1\mathrm{i}\int_0^t\mathrm{Im}(\langle{z_j(\tau), z_c(\tau)}\rangle)d\tau}, \\
& \mbox{i.e.,} \quad z_j(t)=z_j(0) e^{2\kappa_1\mathrm{i}\int_0^t\mathrm{Im}(\langle{z_j(\tau), z_c(\tau)}\rangle)d\tau}.
\end{aligned}
\end{align}
Now, we set
\begin{equation} \label{B-12}
\theta_j(t):= 2\kappa_1 \int_0^t\mathrm{Im}(\langle{z_j(\tau), z_c(\tau)}\rangle) d\tau, \quad j = 1, \cdots, N.
\end{equation}
Finally, we combine \eqref{B-11-2} and \eqref{B-12} to complete the proof of the first assertion. \newline

\noindent (ii)~From \eqref{B-12}, it is easy to see that $\theta_j(0) = 0$. Now we need to check that $\theta_j$ satisfies the Kuramoto type model \eqref{B-11-1}. For this, we substitute the ansatz \eqref{B-11} into $\eqref{C-0}_2$ to get
\[
e^{\mathrm{i}\theta_j}\mathrm{i}\dot{\theta_j}z_j^{in} =\frac{\kappa_1}{N}\sum_{k=1}^N(\langle{z_j^{in}e^{\mathrm{i}\theta_j}, z_k^{in} e^{\mathrm{i}\theta_k}}\rangle-\langle{z_k^{in} e^{\mathrm{i}\theta_k}, z_j^{in} e^{\mathrm{i}\theta_j}}\rangle)e^{\mathrm{i}\theta_j}z_j^{in}.
\]
This yields
\begin{align}
\begin{aligned} \label{B-13}
\mathrm{i}\dot{\theta_j}&=\frac{\kappa_1}{N}\sum_{k=1}^N(\langle{z_j^{in} e^{\mathrm{i}\theta_j}, z_k^{in} e^{\mathrm{i}\theta_k}}\rangle-\langle{z_k^{in} e^{\mathrm{i}\theta_k}, z_j^{in} e^{\mathrm{i}\theta_j}}\rangle) \\
&=\frac{\kappa_1}{N}\sum_{k=1}^N\Big[ 2\mathrm{i}\mathrm{Im}\left(\langle z_j^{in}, z_k^{in} \rangle  e^{\mathrm{i}(\theta_k-\theta_j)} \right) \Big].
\end{aligned}
\end{align}
In \eqref{B-13}, we substitute the ansatz \eqref{B-11-0} and compare the imaginary part of the resulting relation to find the desired system \eqref{B-11-1}.
\end{proof}
In next proposition, we show that the frequency in \eqref{B-11} tends to zero asymptotically. 
\begin{proposition} \label{P3.2}
Let $\{ z_j \}$ be a global solution to $\eqref{C-0}_2$ with the initial data $\{z_j^{in} \}$ whose ansatz is given by \eqref{B-11}. Then, the frequencies ${\dot \theta}_i$ tends to zero asymptotically:
\[   \lim_{t \to \infty} {\dot \theta}_i(t) = 0, \quad i = 1, \cdots, N. \]
\end{proposition}
\begin{proof} Note that $z_c$ satisfies 
\begin{equation} \label{NN-1}
\dot{z}_c=\frac{\kappa_1}{N}\sum_{j=1}^N(\langle{z_j, z_c}\rangle-\langle{z_c, z_j}\rangle)z_j, \quad \mbox{or equivalently} \quad \frac{\dot{z}_j}{z_j}=2\kappa_1\mathrm{i}\mathrm{Im}(\langle{z_j, z_c}\rangle).
\end{equation}
Then, we use \eqref{NN-1} to see
\begin{align}
\begin{aligned} \label{NN-2}
\frac{d}{dt}\langle{z_c, z_c}\rangle&=\langle{\dot{z}_c, z_c}\rangle+\langle{z_c,\dot{z}_c}\rangle\\
&=\frac{\kappa_1}{N}\sum_{j=1}^N(2|\langle{z_j, z_c}\rangle|^2-\langle{z_j, z_c}\rangle^2-\langle{z_c, z_j}\rangle^2)=\frac{4\kappa_1}{N}\sum_{j=1}^N\mathrm{Im}(\langle{z_j, z_c}\rangle)^2.
\end{aligned}
\end{align}
Finally, we combine $\eqref{NN-1}_2$ and \eqref{NN-2} to obtain
\[
\frac{d}{dt}\|z_c\|^2=\frac{4\kappa_1}{N}\sum_{j=1}^N\left(\frac{\dot{z}_j}{2\kappa_1\mathrm{i}z_j}\right)^2=\frac{1}{N\kappa_1}\sum_{j=1}^N\dot{\theta}_j^2.
\]
Since $\|z_c\|$ is non-decreasing and bounded, $\sum_{j=1}^N\dot{\theta}_j^2$ must converges to zero as $t \to \infty$. 
\end{proof}

\subsubsection{A gradient flow formulation} \label{sec:3.2.2} In this part, we study a gradient flow formulation of \eqref{B-11-1}, and study basic properties.
\begin{lemma}  \label{L3.2}
Let $\{\theta_j \}$ be a solution to system \eqref{B-11-1}. Then, we have
\[
\sum_{k=1}^N\theta_k=0 \quad \mbox{and} \quad \frac{d}{dt} \sum_{k=1}^N \theta_k=0. 
\]
\end{lemma}

\begin{proof} (i) We first show the second equality as follows.
\begin{align}
\begin{aligned} \label{B-14}
\sum_{j=1}^N\dot{\theta}_j&=\frac{2\kappa_1}{N}\sum_{j=1}^N\sum_{k=1}^N R^{in}_{jk}\sin(\theta_k-\theta_j+\alpha_{jk})\\
&=\frac{\kappa_1}{N}\sum_{j=1}^N\sum_{k=1}^N R^{in}_{jk}\sin(\theta_k-\theta_j+\alpha_{jk})+\frac{\kappa_1}{N}\sum_{k=1}^N\sum_{j=1}^N R^{in}_{kj}\sin(\theta_j-\theta_k+\alpha_{kj})\\
&=\frac{\kappa_1}{N}\sum_{j=1}^N\sum_{k=1}^N R^{in}_{jk} \Big(\sin(\theta_k-\theta_j+\alpha_{jk})-\sin(\theta_k-\theta_j+\alpha_{jk}) \Big )\\
&=0.
\end{aligned}
\end{align}
(ii) We integrate the relation \eqref{B-14} in $t$ and use $\eqref{B-11-1}_2$ to get 
\[
\sum_{j=1}^N\theta_j(t) = \sum_{j=1}^N\theta_j(0) = 0.
\]
\end{proof}
Next, we provide a gradient flow formulation of system \eqref{B-11-1}. For this, we introduce the potential $V[\Theta]=V(\theta_1(t), \theta_2(t),\cdots, \theta_N(t))$:
\begin{equation} \label{B-15}
V[\Theta] := \frac{\kappa_1}{N}\sum_{i, j = 1}^{N} R^{in}_{ij} \Big(1- \cos(\theta_i-\theta_j+\alpha_{ji}) \Big).
\end{equation}
Note that the potential $V$ is analytic and bounded:
{{\begin{equation} \label{B-15-1}
|V[\Theta]| \leq \Big| \frac{\kappa_1}{N}\sum_{i, j = 1}^{N} R^{in}_{ij} \Big(1- \cos(\theta_i-\theta_j+\alpha_{ji}) \Big) \Big| \leq 2 \frac{\kappa_1}{N}   \sum_{i, j = 1}^{N} |R^{in}_{ij}|.
\end{equation}}}
Moreover, the potential $V[\Theta]$ can be rewritten in terms of $z_j$:
\begin{align*}
V[\Theta]&=-\frac{\kappa_1}{N}\sum_{i, j=1}^N R^{in}_{ij}  \cos(\theta_i-\theta_j+\alpha_{ji}) +  \frac{\kappa_1}{N}\sum_{i, j = 1}^{N} R^{in}_{ij}  \\
&=-\frac{\kappa_1}{N}\sum_{i, j=1}^N\mathrm{Re}\left(R^{in}_{ij}  e^{\mathrm{i}(\theta_i-\theta_j+\alpha_{ji})}\right) +  \frac{\kappa_1}{N}\sum_{i, j = 1}^{N} R^{in}_{ij} \\
&=-\frac{\kappa_1}{N}\sum_{i, j=1}^N\mathrm{Re}\left(\langle{z_j^{in}, z_i^{in}}\rangle e^{\mathrm{i}(\theta_i-\theta_j)}\right) +  \frac{\kappa_1}{N}\sum_{i, j = 1}^{N} R^{in}_{ij} \\
&=-\frac{\kappa_1}{N}\sum_{i, j=1}^N\mathrm{Re}\left(\langle{z_j(t), z_i(t)}\rangle\right) +  \frac{\kappa_1}{N}\sum_{i, j = 1}^{N} R^{in}_{ij} \\
& = {{-\kappa_1 N \|z_c\|^2}}  +  \frac{\kappa_1}{N}\sum_{i, j = 1}^{N} R^{in}_{ij}.
\end{align*}
{Thus, the minimization of $V[\Theta]$ is equal to maximize $\|z_c \|^2$ which has an upper bound $1$.}


\begin{proposition} \label{P3.3}
System $\eqref{B-11-1}_1$ is a gradient flow with the analytical potential $V[\Theta]$ in \eqref{B-15}:
\[ \dot{\Theta}=-\nabla_{\Theta} V[\Theta], \quad t > 0. \]
\end{proposition}
\begin{proof} By straightforward calculation, one has 
\begin{align*}
\begin{aligned}
\partial_{\theta_k} V[\Theta] &=-\frac{\kappa_1}{N} \partial_{\theta_k} \sum_{i, j=1}^N R^{in}_{ij}  \cos(\theta_i-\theta_j+\alpha_{ji})\\
&=\frac{\kappa_1}{N}\sum_{i, j=1}^N R^{in}_{ij}  \sin(\theta_i-\theta_j+\alpha_{ji})\left(\frac{\partial\theta_i}{\partial\theta_k}-\frac{\partial \theta_j}{\partial \theta_k}\right)\\
&=\frac{\kappa_1}{N}\sum_{i, j=1}^N R^{in}_{ij}  \sin(\theta_i-\theta_j+\alpha_{ji})\left(\delta_{ik}-\delta_{jk}\right)\\
&=\frac{\kappa_1}{N}\left(\sum_{j=1}^NR^{in}_{kj}\sin(\theta_k-\theta_j+\alpha_{jk})-\sum_{i=1}^NR^{in}_{ik}\sin(\theta_i-\theta_k+\alpha_{ki})\right)\\
&=\frac{2\kappa_1}{N}\sum_{j=1}^N R^{in}_{kj}\sin(\theta_k-\theta_j+\alpha_{jk}) =-\frac{2\kappa_1}{N}\sum_{j=1}^NR^{in}_{jk}\sin(\theta_j-\theta_k+\alpha_{kj}) \\
&=-\dot{\theta}_k.
\end{aligned}
\end{align*}
This yields the desired estimate.
\end{proof}
\begin{remark}
1. Let $\{ z_j \}$ be a global solution to $\eqref{C-0}_2$. Then, it is easy to see that 
\begin{align}\label{I-2}
\dot{V}[\Theta] =\nabla{V}_{\Theta} \cdot \dot{\Theta} =-\|\nabla_{\Theta}{V}[\Theta]\|^2\leq0.
\end{align}
Thus, {the potential $V[\Theta(t)]$ is bounded (see \eqref{B-15-1}) and decreases over time. Hence, there exists a limit $V_\infty$ such that 
\[ \lim_{t \to \infty} V[\Theta(t)] = V_\infty. \] }
On the other hand, we can easily show that $\dot{\Theta}$ and $\ddot{\Theta}$ are bounded. With these estimates, It follows from Barbalat's lemma that $\dot{\Theta}(t)$ also converges to 0. So we can obtain following limits:
\begin{equation*} \label{I-2-1}
\lim_{t\rightarrow\infty}\dot{\Theta}(t)=0, \quad \lim_{t\rightarrow\infty}\nabla_{\Theta}{V}(\Theta(t))=0.
\end{equation*}
\end{remark}
%
%

{Note that we have studied emergent dynamics of Subsystem A and Subsystem B which correspond to the special cases of \eqref{M-5}. In next section, we consider the dynamic interplay between Subsystem A and Subsystem B with $\kappa_0 > 0$ and $\kappa_1 > 0$ and how the collective behaviors emerge from such interplay.}

\section{The Lohe hermitian sphere model} \label{sec:4}
\setcounter{equation}{0}
In this section, we study emergent dynamics and uniform $\ell^p$-stability of the Lohe hermitian sphere model. 

\subsection{Complete aggregation} \label{sec:4.1} Consider the Lohe hermitian sphere model  for a homogeneous ensemble with $\Omega_j = 0$:
\begin{equation} \label{D-1}
\begin{cases}
\dot{z}_j =  \kappa_{0} \Big(z_c\langle z_j, z_j \rangle-z_j \langle z_c. z_j \rangle \Big)+\kappa_1 \Big(\langle{z_j, z_c}\rangle- \langle z_c, z_j\rangle \Big) z_j.\\
z_j(0)=z_j^{in},\quad \|z_j^{in}\|=1\quad \forall j=1, 2, \cdots, N.
\end{cases}
\end{equation}
First, we consider the time-evolution of the two-point correlation function:
\begin{align}
\begin{aligned} \label{D-3}
\frac{d}{dt}\langle{z_i,z_j}\rangle &=\langle{\dot{x}_i, z_j}\rangle+\langle{z_i, \dot{x}_j}\rangle\\
&=\kappa_0(\langle{z_i, z_i}\rangle \langle{z_c, z_j}\rangle-\langle{z_i,z_c}\rangle \langle{z_i, z_j}\rangle)+\kappa_1(\langle{z_c, z_i}\rangle-\langle{z_i, z_c}\rangle)\langle{z_i, z_j}\rangle\\
&\hspace{0.3cm} +\kappa_0(\langle{z_j, z_j}\rangle \langle{z_i, z_c}\rangle-\langle{z_c,z_j}\rangle \langle{z_i, z_j}\rangle)+\kappa_1(\langle{z_j, z_c}\rangle-\langle{z_c, z_j}\rangle)\langle{z_i, z_j}\rangle\\
&=\kappa_0(1-\langle{z_i, z_j}\rangle)(\langle{z_i, z_c}\rangle+\langle{z_c, z_j}\rangle) \\
&\hspace{0.3cm} +\kappa_1\langle{z_i, z_j}\rangle(\langle{z_c, z_i}\rangle-\langle{z_i, z_c}\rangle+\langle{z_j, z_c}\rangle-\langle{z_c, z_j}\rangle).
\end{aligned}
\end{align}
\begin{lemma} \label{L4.1}
Let $\{z_j \}$ be a global solution to \eqref{D-1} with $\rho^{in}=\|z_c^{in}\|>0$. Then, $\rho = \| z_c \|$ satisfies
\[
\frac{d\rho^2}{dt}=\frac{2\kappa_0}{N}\sum_{i = 1}^{N} \Big( \rho^2-|\langle{z_i, z_c}\rangle|^2 \Big) +\frac{4(\kappa_0+\kappa_1)}{N}\sum_{i =1}^{N} \Big| \mathrm{Im}(\langle{z_i, z_c}\rangle) \Big|^2 \geq 0.
\]
\end{lemma}
\begin{proof} It follows from \eqref{D-3} that 
\begin{align}
\begin{aligned} \label{D-4}
\frac{d}{dt}\sum_{i, j = 1}^{N} \langle{z_i, z_j}\rangle &=  \kappa_0 \sum_{i, j =1}^{N} \Big(1-\langle{z_i, z_j}\rangle)(\langle{z_i, z_c}\rangle+\langle{z_c, z_j}\rangle \Big) \\
&\hspace{0.2cm} + \kappa_1 \sum_{i, j = 1}^{N} \langle{z_i, z_j}\rangle \Big (\langle{z_c, z_i}\rangle-\langle{z_i, z_c}\rangle+\langle{z_j, z_c}\rangle-\langle{z_c, z_j}\rangle \Big )\\
&=N \kappa_0 \sum_{i=1}^{N}  \Big (\langle{z_c, z_c}\rangle+\langle{z_c, z_c}\rangle-\langle{z_i, z_c}\rangle^2-\langle{z_c, z_i}\rangle^2 \Big) \\
&\hspace{0.2cm} +N \kappa_1 \sum_{i = 1}^{N}  \Big (2|\langle{z_i, z_c}\rangle|^2-\langle{z_i, z_c}\rangle^2-\langle{z_c, z_i}\rangle^2 \Big).
\end{aligned}
\end{align}
This implies
\begin{align*}
\frac{d\rho^2}{dt}&=\frac{\kappa_0}{N}\sum_{i=1}^{N} \left(2\|z_c\|^2-\langle{z_i, z_c}\rangle^2-\langle{z_c, z_i}\rangle^2\right) +\frac{\kappa_1}{N}\sum_{i = 1}^{N} (2|\langle{z_i, z_c}\rangle|^2-2\mathrm{Re}(\langle{z_i, z_c}\rangle^2))\\
&=\frac{2\kappa_0}{N}\sum_{i = 1}^{N} \left(\|z_c\|^2-\mathrm{Re}\left(\langle{z_i, z_c}\rangle^2\right)\right)+\frac{2\kappa_1}{N}\sum_{i=1}^{N} (|\langle{z_i, z_c}\rangle|^2-\mathrm{Re}(\langle{z_i, z_c}\rangle)^2+\mathrm{Im}(\langle{z_i, z_c}\rangle)^2)\\
&=\frac{2\kappa_0}{N}\sum_{i = 1}^{N} \left(\|z_c\|^2-\mathrm{Re}(\langle{z_i, z_c}\rangle)^2+\mathrm{Im}(\langle{z_i, z_c}\rangle)^2\right)+\frac{4\kappa_1}{N}\sum_{i=1}^{N} \mathrm{Im}(\langle{z_i, z_c}\rangle)^2\\
&=\frac{2\kappa_0}{N}\sum_{i =1}^{N} \left(\|z_c\|^2-|\langle{z_i, z_c}\rangle|^2\right)+\frac{4(\kappa_0+\kappa_1)}{N}\sum_{i =1}^{N} \Big| \mathrm{Im}(\langle{z_i, z_c}\rangle) \Big|^2.
\end{align*}
\end{proof}
\noindent {First, note that $\rho$ is uniformly bounded by 1. It follows from Lemma \ref{L2.1} that 
\[ \|z_j(t) \| = \|z_j^{in} \| = 1. \]
This yields a uniform boundedness of $\rho$:
\[ \rho(t) = \|z_c(t) \| \leq \frac{1}{N} \sum_{j=1}^{N} \|z_j(t) \| \leq 1. \]
On the other hand, since $\rho$ is non-decreasing along the dynamics \eqref{D-1}, $\rho$ converges to some value in $[0, 1]$ asymptotically}. From this property, we can obtain the following corollary.
\begin{corollary} \label{C4.1}
Let $\{z_j \}$ be a global solution to \eqref{D-1} with $\rho^{in}>0$. Then,  
\begin{align*}
\begin{aligned}
& (i)~\exists~\rho^\infty>0\mbox{ such that }\lim_{t\rightarrow\infty} \rho(t) =\rho^\infty. \\
& (ii)~\frac{2\kappa_0}{N}\sum_{i=1}^N\int_0^\infty(\rho(s)^2-|\langle{z_i(s), z_c(s)}\rangle|^2)ds \\
& \hspace{2cm} +\frac{4(\kappa_0+\kappa_1)}{N}\sum_{i=1}^N\int_0^\infty\left|\mathrm{Im}(\langle{z_i(s), z_c(s)}\rangle)\right|^2ds\leq1-(\rho^{in})^2.\\
& (iii)~\lim_{t\rightarrow\infty}(\|z_c\|^2-|\langle{z_i, z_c}\rangle|^2)=0,\quad\lim_{t\rightarrow\infty}\mathrm{Im}(\langle{z_i, z_c}\rangle)=0. \\
& (iv)~\lim_{t\rightarrow\infty}\langle{z_i, z_c}\rangle  \in \{1, -1\} \quad \forall~i=1, \cdots, N.
\end{aligned}
\end{align*}
\end{corollary}
\begin{proof} Below, we provide proofs for each assertion separately. \newline

\noindent (i)~Since $\frac{d}{dt}\rho^2$ is increasing and bounded, there exist a nonnegative number $\rho^\infty$ such that
\[
\lim_{t\rightarrow\infty}\rho(t)=\rho^\infty.
\]

\noindent (ii)~We integrate the equation in lemma \ref{L4.1} to obtain
\begin{align*}
\begin{aligned}
\rho(t)^2-(\rho^{in})^2 &=\frac{2\kappa_0}{N}\sum_{i=1}^N\int_0^t(\rho(s)^2-|\langle{z_i(s), z_c(s)}\rangle|^2)ds \\
&\hspace{1cm} +\frac{4(\kappa_0+\kappa_1)}{N}\sum_{i=1}^N\int_0^t\left|\mathrm{Im}(\langle{z_i(s), z_c(s)}\rangle)\right|^2ds.
\end{aligned}
\end{align*}
This yields
\begin{align*}
\begin{aligned}
1-(\rho^{in})^2&\geq(\rho^\infty)^2-(\rho^{in})^2 =\frac{2\kappa_0}{N}\sum_{i=1}^N\int_0^\infty(\rho(s)^2-|\langle{z_i(s), z_c(s)}\rangle|^2)ds \\
&+\frac{4(\kappa_0+\kappa_1)}{N}\sum_{i=1}^N\int_0^\infty\left|\mathrm{Im}(\langle{z_i(s), z_c(s)}\rangle)\right|^2ds.
\end{aligned}
\end{align*}
\noindent (iii)~It follows from Barbalat's lemma and  uniform boundedness of $\frac{d^2 }{dt^2}\rho^2$ that
\[
\lim_{t\rightarrow\infty}\frac{d}{dt}\rho(t)^2=0.
\]
Moreover, it follows from Lemma \ref{L4.1} that 
\[
\frac{2\kappa_0}{N}\sum_{i = 1}^{N} \Big( \rho^2-|\langle{z_i, z_c}\rangle|^2 \Big) +\frac{4(\kappa_0+\kappa_1)}{N}\sum_{i =1}^{N} \Big| \mathrm{Im}(\langle{z_i, z_c}\rangle) \Big|^2 \to 0, \quad \mbox{as $t \to \infty$}.
\]
Since each term is nonnegative, we obtain
\[
\lim_{t\rightarrow\infty}(\rho(t)^2-|\langle{z_i, z_c}\rangle|^2)=0\qquad \mbox{and} \qquad \lim_{t\rightarrow\infty}\mathrm{Im}(\langle{z_i(t), z_c(t)}\rangle)=0.
\]
\noindent (iv)~By (ii), we can see that $\frac{dz_i}{dt}$ converges absolutely. So there exists
\[
\lim_{t\rightarrow\infty}z_i(t)=z_i^\infty.
\]
This implies
\[
\|z_c^\infty\|^2=|\langle{z_i^\infty, z_c^\infty}\rangle|^2 \qquad\mbox{and}\qquad \mathrm{Im}(\langle{z_i^\infty, z_c^\infty}\rangle)=0.
\]
On the other hand, it follows from Cauchy-Schwarz inequality that 
\[
\|z^{\infty}_c \|^2 = \|z_i^\infty\|^2\|z_c^\infty\|^2\geq|\langle{z_i^\infty, z_c^\infty}\rangle|^2.
\]
Since the equality holds, we can obtain 
\[
z_i^\infty=\alpha_iz_c^\infty.
\]
If we put above relation in $\mathrm{Im}(\langle{z_i^\infty, z_c^\infty}\rangle)=0$, we obtain
\[
\mathrm{Im}(\alpha_i)=0.
\]
So $\alpha_i$ is non-zero real numbers. Thus, we can obtain that each $z_i^\infty$ has two clusters  with $\alpha_i > 0$ or $\alpha_i < 0$.
\end{proof}

\begin{remark}
\begin{enumerate}
\item Let $\{z_j \}$ be a global solution to \eqref{D-1} with  $\rho^{in}>\frac{N-2}{N}$. Then, we have
\[
\lim_{t\rightarrow\infty}\rho(t)=1.
\]
If each clusters contain $l$ and $N-l$ particles, then we have $\rho=\frac{|N-2l|}{N}$. So if $\rho^{in}>\frac{N-2}{N}$, we can obtain 
\[ l=0 \quad \mbox{or} \quad l=N. \]
That means there is only one cluster. i.e. complete aggregation.\\

\item Let $\{z_j \}$ be a global solution to \eqref{D-1} with $\rho^{in}>0$. Then, it follows from Lemma \ref{L4.1} that $\rho$ is increasing along the flow \eqref{D-1}. Thus, there will be no nontrivial periodic solution. 
\end{enumerate}
\end{remark}

Next, we introduce a Lyapunov functional ${\mathcal L}(Z)$ with $Z=(z_1, z_2, \cdots, z_N)$:
\[ \mathcal{L}(Z) :=\max_{1\leq i, j \leq N}|1-\langle{z_i, z_j}\rangle|^2, \]
and study its time-evolution. It follows from \eqref{D-4} that  
\begin{align*}
\begin{aligned}
\frac{d}{dt}(1-\langle{z_i,z_j}\rangle) &=-\kappa_0(1-\langle{z_i, z_j}\rangle)(\langle{z_i, z_c}\rangle+\langle{z_c, z_j}\rangle) \\
&\hspace{0.2cm} -\kappa_1\langle{z_i, z_j}\rangle(\langle{z_c, z_i}\rangle-\langle{z_i, z_c}\rangle+\langle{z_j, z_c}\rangle-\langle{z_c, z_j}\rangle).
\end{aligned}
\end{align*}
This yields
\begin{align}
\begin{aligned} \label{D-5}
&\frac{d}{dt}|1-\langle{z_i, z_j}\rangle|^2 \\
& \hspace{0.5cm} =-\kappa_0|1-\langle{z_i, z_j}\rangle|^2(\langle{z_i, z_c}\rangle+\langle{z_c, z_j}\rangle) \\
& \hspace{0.7cm} -\kappa_1\langle{z_i, z_j}\rangle(1-\langle{z_j, z_i}\rangle)(\langle{z_c, z_i}\rangle-\langle{z_i, z_c}\rangle+\langle{z_j, z_c}\rangle-\langle{z_c, z_j}\rangle)\\
& \hspace{0.7cm} -\kappa_0|1-\langle{z_i, z_j}\rangle|^2(\langle{z_c, z_i}\rangle+\langle{z_j, z_c}\rangle) \\
& \hspace{0.7cm} +\kappa_1\langle{z_j, z_i}\rangle(1-\langle{z_i, z_j}\rangle)(\langle{z_c, z_i}\rangle-\langle{z_i, z_c}\rangle+\langle{z_j, z_c}\rangle-\langle{z_c, z_j}\rangle)\\
& \hspace{0.5cm} =-\kappa_0|1-\langle{z_i, z_j}\rangle|^2(\langle{z_i, z_c}\rangle+\langle{z_c, z_j}\rangle+\langle{z_c, z_i}\rangle+\langle{z_j, z_c}\rangle)\\
& \hspace{0.7cm} -\kappa_1(\langle{z_i, z_j}\rangle-\langle{z_j, z_i})(\langle{z_c, z_i}\rangle-\langle{z_i, z_c}\rangle+\langle{z_j, z_c}\rangle-\langle{z_c, z_j}\rangle)\\
& \hspace{0.5cm} =-2\kappa_0|1-\langle{z_i, z_j}\rangle|^2\mathrm{Re}(\langle{z_i+z_j, z_c}\rangle)+4\kappa_1\mathrm{Im}(\langle{z_i, z_j}\rangle)\mathrm{Im}(\langle{z_c, z_i-z_j}\rangle).
\end{aligned}
\end{align}
Now, we choose indices $i_0$ and $j_0$ such that 
\begin{equation} \label{D-6}
\mathcal{L}(Z) =:|1-\langle{z_{i_0}, z_{j_0}}\rangle|^2.
\end{equation}
Note that in general, system \eqref{D-1} implies
\[ |1-\langle{z_i, z_j}\rangle|^2\neq|z_i-z_j|^2, \quad \mbox{for some $i, j$}. \]
It follows from \eqref{D-5} and \eqref{D-6} that  
\[
\frac{d}{dt}\mathcal{L}(Z)^2 \leq-2\kappa_0 \mathcal{L}(Z)^2\mathrm{Re}(\langle{x_{i_0}+x_{j_0}}, z_c\rangle)+8\kappa_1\mathcal{L}(Z)^2. 
\]
This yields
\[
\frac{d}{dt}\mathcal{L}(Z)\leq-\kappa_0\mathcal{L}(Z)\left(\mathrm{Re}(\langle{x_{i_0}+x_{j_0}, z_c}\rangle)-\frac{4\kappa_1}{\kappa_0}\right).
\]
We know that $\langle{z_i, z_c}\rangle$ converges to 1 for all $i$ under some condition. Finally, we can obtain following theorem.

\begin{theorem} \label{T4.1}
Suppose that the coupling strengths and initial data satisfy
\[ 0< \kappa_1 <  \frac{1}{4} \kappa_0, \quad \rho^{in} > \frac{N-2}{N},  \]
and let $\{z_j \}$ be a global solution to \eqref{D-1} with $\Omega_j = 0$. Then $\mathcal{D}(X)$ converges to 0 exponentially fast.
\end{theorem}
\begin{proof}
Since $\langle{z_i, z_c}\rangle$ converges to 1 for all $i$ as $t \to \infty$, we can obtain that  for some $T, \varepsilon>0$
\[
\mathrm{Re}(\langle{z_{i}(t)+z_{j}(t), z_c}\rangle)-\frac{4\kappa_1}{\kappa_0}>\varepsilon, \quad \mbox{for all $t > 0$}.
\]
 This yields
\[
\frac{1}{\mathcal{L}(Z(t))}\frac{d}{dt}{\mathcal{L}(Z(t))}\leq-\kappa_0 \varepsilon,
\]
for all $t>T$. Hence, we can conclude that
\[
\mathcal{L}(Z(t))<C\exp(-\kappa_0\varepsilon t)
\]
for some positive $C>0$. 
\end{proof}
\begin{remark} \label{R4.2}
1. For the special case with $\kappa_1 = -\kappa_0$, system \eqref{D-1} is a gradient flow. More precisely, note that 
\begin{align}
\begin{aligned} \label{D-7}
\dot{z}_j&= \kappa_0(\langle{z_j, z_j}\rangle z_c-\langle{z_c, z_j}\rangle z_j)+\kappa_1(\langle{z_j, z_c}\rangle-\langle{z_c, z_j}\rangle)z_j\\
&= \kappa_0( z_c-\langle{z_c, z_j}\rangle z_j)+\kappa_1(\langle{z_j, z_c}\rangle-\langle{z_c, z_j}\rangle)z_j\\
&= \kappa_0(z_c-\langle{z_j, z_c}\rangle z_j)+(\kappa_0+\kappa_1)(\langle{z_j, z_c}\rangle-\langle{z_c, z_j}\rangle)z_j\\
&= \kappa_0 \mathbb{P}^{\perp}_{z_j}z_c+(\kappa_0+\kappa_1)(\langle{z_j, z_c}\rangle-\langle{z_c, z_j}\rangle)z_j\\
&= \kappa_0 \mathbb{P}^{\perp}_{z_j}z_c+2(\kappa_0+\kappa_1)\mathrm{Im}(\langle{z_j, z_c}\rangle))z_j,
\end{aligned}
\end{align}
{where $\mathbb{P}^\perp$ is orthogonal projection of $z_c$ onto the tangent plane of hermitian unit sphere at $z_j$.} Hence, for $\kappa_0 + \kappa_1 = 0$, relation \eqref{D-7} shows 
\[ \dot{z}_j=\kappa_0\mathbb{P}^{\perp}_{z_j}z_c. \]

\noindent {2.~Emergent dynamics for the Lohe sphere model has been mostly studied on the complete network. However, there are very few literature \cite{Z-Z-Q, Z-Z} dealing with emergent dynamics of the Lohe sphere model over non all-to-all networks in which equilibria and consensus were studied for identical oscillators.}
\end{remark}

\subsection{Uniform $\ell^p$-stability} \label{sec:4.2}
In this subsection, we study the uniform $\ell^p$-stability of \eqref{D-1} with respect to initial data. For this, let $Z = \{ z_j \}$ and ${\tilde Z} = \{ {\tilde z}_j \}$ be two global solutions to \eqref{D-1}. Then, they satisfy 
\begin{align}
\begin{aligned} \label{D-8}
\dot{z}_j &= \frac{\kappa_0}{N}\sum_{k=1}^N(z_k-h_{kj}z_j)+\frac{\kappa_1}{N}\sum_{k=1}^N(h_{jk}-h_{kj})z_j,\\
\dot{\tilde{z}}_j&= \frac{\kappa_0}{N}\sum_{k=1}^N(\tilde{z}_k-\tilde{h}_{kj}\tilde{z}_j)+\frac{\kappa_1}{N}\sum_{k=1}^N(\tilde{h}_{jk}-\tilde{h}_{kj})\tilde{z}_j,
\end{aligned}
\end{align}
where $h_{ij} = \langle z_i, z_j \rangle$, $\tilde{h}_{ij}=\langle \tilde{z}_i, \tilde{z}_j\rangle$. \newline

\noindent It follows from \eqref{D-8} that
\begin{align}
\begin{aligned} \label{D-9}
\frac{d}{dt}(\tilde{z}_j-z_j) &= \frac{\kappa_0}{N}\sum_{k=1}^N[(\tilde{z}_k-z_k)-(\tilde{h}_{kj}\tilde{z}_j-h_{kj}z_j)]  \\
&+\frac{\kappa_1}{N}\sum_{k=1}^N[(\tilde{h}_{jk}\tilde{z}_j-h_{jk}z_j)-(\tilde{h}_{kj}\tilde{z}_j-h_{kj}z_j)].
\end{aligned}
\end{align}
We take an inner product \eqref{D-9} with $\tilde{z}_j-z_j$ to get
\begin{align}
\begin{aligned} \label{B-19}
&\langle{\dot{\tilde{z}}_j-\dot{z}_j, \tilde{z}_j-z_j}\rangle+\langle{{\tilde{z}}_j-{z}_j, \dot{\tilde{z}}_j-\dot{z}_j}\rangle\\
& \hspace{1cm} =\frac{\kappa_0}{N}\sum_{k=1}^N\langle{(\tilde{z}_k-z_k)-(\tilde{h}_{kj}\tilde{z}_j-h_{kj}z_j), \tilde{z}_j-z_j}\rangle  \\
&  \hspace{1.2cm}+\frac{\kappa_0}{N}\sum_{k=1}^N\langle{\tilde{z}_j-z_j, (\tilde{z}_k-z_k)-(\tilde{h}_{kj}\tilde{z}_j-h_{kj}z_j)}\rangle)\\
& \hspace{1.2cm}+\frac{\kappa_1}{N}\sum_{k=1}^N\langle{(\tilde{h}_{jk}\tilde{z}_j-h_{jk}z_j)-(\tilde{h}_{kj}\tilde{z}_j-h_{kj}z_j), \tilde{z}_j-z_j }\rangle \\
& \hspace{1.2cm} +\frac{\kappa_1}{N}\sum_{k=1}^N\langle{\tilde{z}_j-z_j, (\tilde{h}_{jk}\tilde{z}_j-h_{jk}z_j)-(\tilde{h}_{kj}\tilde{z}_j-h_{kj}z_j)}\rangle.
\end{aligned}
\end{align}
The last two terms involving with $\kappa_1$ in \eqref{B-19} can be reduced as follows.
\begin{align*}
&\langle{(\tilde{h}_{jk}\tilde{z}_j-h_{jk}z_j)-(\tilde{h}_{kj}\tilde{z}_j-h_{kj}z_j), \tilde{z}_j-z_j }\rangle+\langle{\tilde{z}_j-z_j, (\tilde{h}_{jk}\tilde{z}_j-h_{jk}z_j)-(\tilde{h}_{kj}\tilde{z}_j-h_{kj}z_j)}\rangle\\
& \hspace{0.2cm} =\langle{\tilde{z}_j-z_j, (\tilde{h}_{jk}\tilde{z}_j-h_{jk}z_j)-(\tilde{h}_{kj}\tilde{z}_j-h_{kj}z_j)}\rangle+\mbox{(c. c.)}\\
&  \hspace{0.2cm}  =\tilde{h}_{jk}-h_{jk}\langle{\tilde{z}_j, z_j}\rangle-\tilde{h}_{kj}+h_{kj}\langle{\tilde{z}_j, z_j}\rangle-\tilde{h}_{jk}\langle{z_j, \tilde{z}_j}\rangle+h_{jk}+\tilde{h}_{kj}\langle{z_j, \tilde{z}_j}\rangle-h_{kj}+\mbox{(c. c.)}\\
& \hspace{0.2cm}  =-h_{jk}\langle{\tilde{z}_j, z_j}\rangle+h_{kj}\langle{\tilde{z}_j, z_j}\rangle-\tilde{h}_{jk}\langle{z_j, \tilde{z}_j}\rangle+\tilde{h}_{kj}\langle{z_j, \tilde{z}_j}\rangle+\mbox{(c. c.)}\\
& \hspace{0.2cm}  =(h_{kj}-h_{jk})\langle{\tilde{z}_j, z_j}\rangle+(\tilde{h}_{kj}-\tilde{h}_{jk})\langle{z_j, \tilde{z}_j}\rangle+\mbox{(c. c.)}\\
& \hspace{0.2cm}  =(h_{kj}-h_{jk}+\tilde{h}_{jk}-\tilde{h}_{kj})(\langle{\tilde{z}_j, z_j}\rangle-\langle{z_j, \tilde{z_j}}\rangle)\\
& \hspace{0.2cm}  =2\mathrm{Im}(h_{kj}-\tilde{h}_{kj})\mathrm{Im}(\langle{z_j, \tilde{z}_j}\rangle),
\end{align*}
where (c. c.) denotes a complex conjugate of the preceding terms.
Thus, one has
\begin{align}
\begin{aligned} \label{B-20}
\frac{d}{dt}\|\tilde{z}_j-z_j\|^2 &=\frac{\kappa_0}{N}\sum_{k=1}^N(\langle{\tilde{z}_k-z_k, \tilde{z}_j-z_j}\rangle+\langle{\tilde{z}_j-z_j, \tilde{z}_k-z_k}\rangle \\
&-\frac{\kappa_0}{N}\sum_{k=1}^N(\langle{\tilde{h}_{kj}(\tilde{z}_j-z_j), \tilde{z}_j-z_j}\rangle
+\langle{\tilde{z}_j-z_j}, \tilde{h}_{kj}(\tilde{z}_j-z_j) \rangle) \\
&+\frac{\kappa_0}{N}\sum_{k=1}^N(\langle{(h_{kj}-\tilde{h}_{kj})z_j, \tilde{z}_j-z_j}\rangle
+\langle{\tilde{z}_j-z_j, (h_{kj}-\tilde{h}_{kj})z_j}\rangle ) \\
&+\frac{2\kappa_1}{N}\sum_{k=1}^N\mathrm{Im}(h_{kj}-\tilde{h}_{kj})\mathrm{Im}(\langle{z_j, \tilde{z}_j}\rangle)\\
& =: {\mathcal I}_{21} + {\mathcal I}_{22} + {\mathcal I}_{23}+{\mathcal I}_{24}.
\end{aligned}
\end{align}
In the following lemma, we provide estimates for ${\mathcal I}_{2i}$.

\begin{lemma} \label{L4.2} The following estimates hold.
\begin{align*}
\begin{aligned}
& (i)~\mathcal{I}_{21} \leq\frac{2\kappa_0}{N}\sum_{k=1}^N|\tilde{z}_k-z_k|\cdot|\tilde{z}_j-z_j|, \quad \mathcal{I}_{22} =-\frac{2\kappa_0}{N}\sum_{k=1}^N\mathrm{Re}(\tilde{h}^{in}_{kj})|\tilde{z}_j-z_j|^2.  \\
& (ii)~\mathcal{I}_{23} =-\frac{\kappa_0}{N}\sum_{k=1}^N[\mathrm{Re}(h^{in}_{kj}-\tilde{h}^{in}_{kj})|z_j-\tilde{z}_j|^2+2\mathrm{Im}(h^{in}_{kj}-\tilde{h}^{in}_{kj})\mathrm{Im}(\langle{ \tilde{z}_j, z_j}\rangle)].
\end{aligned}
\end{align*}
\end{lemma}
\begin{proof} $\bullet$~(Estimates on ${\mathcal I}_{2i},~i=1,2$): By direct estimates, one has 
\begin{align*}
\begin{aligned}
\mathcal{I}_{21} &=\frac{2\kappa_0}{N}\sum_{k=1}^N\mathrm{Re}(\langle{\tilde{z}_k-z_k, \tilde{z}_j-z_j)}\rangle
\leq\frac{2\kappa_0}{N}\sum_{k=1}^N|\tilde{z}_k-z_k|\cdot|\tilde{z}_j-z_j|, \\
\mathcal{I}_{22} &=-\frac{2\kappa_0}{N}\sum_{k=1}^N\mathrm{Re}(\tilde{h}^{in}_{kj})|\tilde{z}_j-z_j|^2.
\end{aligned}
\end{align*}

\vspace{0.2cm}

\noindent $\bullet$~(Estimates on ${\mathcal I}_{23}$): Similar to other terms, one has 
\begin{align*}
\mathcal{I}_{23} &=\frac{2\kappa_0}{N}\sum_{k=1}^N\mathrm{Re}(\langle{(h^{in}_{kj}-\tilde{h}^{in}_{kj})z_j, \tilde{z}_j-z_j}\rangle) =-\frac{2\kappa_0}{N}\sum_{k=1}^N\mathrm{Re}((h^{in}_{kj}-\tilde{h}^{in}_{kj})\langle{ z_j-\tilde{z}_j, z_j}\rangle)\\
&=-\frac{2\kappa_0}{N}\sum_{k=1}^N[\mathrm{Re}(h^{in}_{kj}-\tilde{h}^{in}_{kj})\mathrm{Re}(\langle{ z_j-\tilde{z}_j, z_j}\rangle)-\mathrm{Im}(h^{in}_{kj}-\tilde{h}^{in}_{kj})\mathrm{Im}(\langle{ z_j-\tilde{z}_j, z_j}\rangle)]\\
&=-\frac{\kappa_0}{N}\sum_{k=1}^N[\mathrm{Re}(h^{in}_{kj}-\tilde{h}^{in}_{kj})|z_j-\tilde{z}_j|^2+2\mathrm{Im}(h_{kj}-\tilde{h}_{kj})\mathrm{Im}(\langle{ \tilde{z}_j, z_j}\rangle)].
\end{align*}
\end{proof}
Now, we are ready to provide our uniform stability estimate as follows.
\begin{theorem} \label{T4.2}
Suppose that $\kappa_0$, $\kappa_1$, $p$ and initial data $Z^{in}, \tilde{Z}^{in}$ satisfy
\begin{align*}
\begin{aligned}
& 0 <\kappa_1<\frac{\kappa_0}{4}, \quad  z_i^{in} \neq z_j^{in}, \quad \tilde{z}_i^{in} \neq \tilde{z}_j^{in}, \quad \forall i\neq j, \\
&p\in[1, \infty), \quad   \rho^{in}>\frac{N-2}{N},\quad \tilde{\rho}^{in}>\frac{N-2}{N},
\end{aligned}
\end{align*}
and let $Z=\{z_j \}$ and $\tilde{Z}= \{ {\tilde z}_j \}$ be global solutions to \eqref{D-1} with the initial data $Z^{in}$ and $\tilde{Z}^{in}$, respectively. Then, there exists a constant $G>0$ independent of $t$ such that
\[
\sup_{0\leq t<\infty}\|Z(t)-\tilde{Z}(t)\|_p\leq G\| Z^{in} -\tilde{Z}^{in} \| _p.
\]
\end{theorem}
\begin{proof} For notational simplicity, we set
\[
h_{ij} :=\langle{z_i, z_j}\rangle,\quad \tilde{h}_{ij} :=\langle{\tilde{z}_i, \tilde{z}_j}\rangle.
\]
Then, in \eqref{B-20} one has
\begin{align*}
\begin{aligned}
&\frac{d}{dt}\|\tilde{z}_j-z_j\|^2  =\mathcal{I}_{21} +\mathcal{I}_{22} +\mathcal{I}_{23}+\mathcal{I}_{24} \\
& \hspace{0.5cm} \leq \frac{2\kappa_0}{N}\sum_{k=1}^N|\tilde{z}_k-z_k|\cdot|\tilde{z}_j-z_j|-\frac{2\kappa_0}{N}\sum_{k=1}^N\mathrm{Re}(\tilde{h}_{kj})|\tilde{z}_j-z_j|^2\\
& \hspace{0.7cm} -\frac{\kappa_0}{N}\sum_{k=1}^N[\mathrm{Re}(h_{kj}-\tilde{h}_{kj})|z_j-\tilde{z}_j|^2+2\mathrm{Im}(h_{kj}-\tilde{h}_{kj})\mathrm{Im}(\langle{ \tilde{z}_j, z_j}\rangle)]\\
& \hspace{0.5cm} =\frac{2\kappa_0}{N}\sum_{k=1}^N|\tilde{z}_k-z_k|\cdot|\tilde{z}_j-z_j|-\frac{\kappa_0}{N}\sum_{k=1}^N\mathrm{Re}(\tilde{h}_{kj}+h_{kj})|\tilde{z}_j-z_j|^2\\
& \hspace{0.7cm} -\frac{2\kappa_0}{N}\sum_{k=1}^N\mathrm{Im}(h_{kj}-\tilde{h}_{kj})\mathrm{Im}(\langle{\tilde{z}_j, z_j}\rangle)+\frac{2\kappa_1}{N}\sum_{k=1}^N\mathrm{Im}(h_{kj}-\tilde{h}_{kj})\mathrm{Im}(\langle{z_j, \tilde{z}_j}\rangle)\\
& \hspace{0.5cm} \leq\frac{2\kappa_0}{N}\sum_{k=1}^N|\tilde{z}_k-z_k|\cdot|\tilde{z}_j-z_j|-\frac{\kappa_0}{N}\sum_{k=1}^N\mathrm{Re}(\tilde{h}_{kj}+h_{kj})|\tilde{z}_j-z_j|^2\\
& \hspace{0.7cm} +{(\kappa_0+\kappa_1)}|\mathrm{Im}(h_{cj}-\tilde{h}_{cj})|\cdot|\tilde{z}_j-z_j|^2.
\end{aligned}
\end{align*}
This yields
\begin{align*}
\begin{aligned}
\frac{d}{dt}\sum_{j=1}^N|\tilde{z}_j-z_j|^p &\leq \frac{\kappa_0 p}{N}\sum_{j=1}^N\sum_{k=1}^N|\tilde{z}_k-z_k|\cdot|\tilde{z}_j-z_j|^{p-1} \\
&-\frac{\kappa_0p}{2N}\sum_{j=1}^N\sum_{k=1}^N\mathrm{Re}(\tilde{h}_{kj}+h_{kj})|\tilde{z}_j-z_j|^p\\
&+\frac{(\kappa_0+\kappa_1)p}{2}\sum_{j=1}^N|\mathrm{Im}(h_{cj}-\tilde{h}_{cj})|\cdot|\tilde{z}_j-z_j|^p\\
&\leq \kappa_0 p \|Z-\tilde{Z}\|_p^p-\kappa_0p(1-C_0e^{-D_0 t})\|Z-\tilde{Z}\|_p^p  \\
&+{(\kappa_0+\kappa_1)p}\max_{k}(\max[\mathrm{Im}(h_{ck}), \mathrm{Im}(\tilde{h}_{ck})])\cdot\|\tilde{Z}-Z\|^p_p.
\end{aligned}
\end{align*}
Now, we use 
\[
\mathrm{Re}(h_{ij})^2+\mathrm{Im}(h_{ij})^2\leq1,
\]
and the exponential decay of $\mathrm{Re}(1-h_{ij})$ to 0 in Theorem \ref{T4.1} to see
\[
\mathrm{Re}(1-h_{ij})\leq C_0e^{-D_0t}.
\]
If we substitute above estimate into the first inequality, we obtain
\[
\mathrm{Im}(h_{ij})^2\leq(1-\mathrm{Re}(h_{ij}))(1+\mathrm{Re}(h_{ij}))=2C_0e^{-D_0t}.
\]
So we can obtain
\[
\mathrm{Im}(h_{ij})\leq(1-\mathrm{Re}(h_{ij}))(1+\mathrm{Re}(h_{ij}))=E_0e^{-D_0t/2}\leq E_0 e^{-D_0 t}.
\]
Similarly, we have
\[
\mathrm{Im}(\tilde{h}_{ij})\leq(1-\mathrm{Re}(\tilde{h}_{ij}))(1+\mathrm{Re}(\tilde{h}_{ij}))=E_0e^{-D_0t/2}\leq E_0 e^{-D_0 t}.
\]
Finally, we can obtain 
\[
\frac{d}{dt}\|\tilde{Z}- Z\|_p^p\leq\kappa_0 p(C_0+E_0)e^{-D_0t}\|\tilde{Z}-Z\|_p^p.
\]
By Gronwall's lemma, we establish the uniform $l_p$-stability:
\[
\sup_{0 \leq t < \infty} \|Z(t)-\tilde{Z}(t)\|_p\leq G\|Z^{in}-\tilde{Z}^{in}\|_p.
\]
\end{proof}

\section{Conclusion} \label{sec:5}
\setcounter{equation}{0} 
In this paper, we studied an emergent dynamics of the Lohe hermitian sphere model arising from the Lohe tensor model as a special case for the ensemble of rank-1 tensors. {There might be several ways to lift the Lohe sphere model on $\bbr^d$ to the corresponding model on $\bbc^d$. Then natural question is how to propose such a lifting naturally and logically. For this, we take a top down approach, namely reduction from the Lohe tensor model.} {For real rank-1 tensors, the coupling term involving with $\kappa_1$ in the Lohe hermitian sphere model vanishes due to the commutativity of an inner product, and we recover the Lohe sphere model introduced in earlier literature.}
Interestingly, the Lohe hermitian sphere model includes the Lohe couplings and Kuramoto coupling together. In fact, when the complex-valued tensor collapses to the real-valued tensor, the cubic coupling terms producing Kuramoto dynamics disappear and only the Lohe coupling terms survive. For the proposed Lohe hermitian sphere model, we provide a conserved quantity and emergent dynamics of sub-systems involving with a single coupling term in terms of system parameters and initial data. {More precisely, our first result is an existence of non-trivial conserved quantity, namely the cross-ratio like quantity:  for a configuration $Z = (z_1, \cdots, z_N)$ with $z_i \not = z_j$ for $i \not = j$, 
\[  \frac{d}{dt} \frac{(1-\langle{z_i, z_j}\rangle)(1-\langle{z_k, z_l}\rangle)}{(1-\langle{z_i, z_l}\rangle)(1-\langle{z_k, z_j}\rangle)} = 0,
 \quad t > 0. \]
Our second results deal with the following two subsystems:
\[ \dot{z}_j= \kappa_0(\langle{z_j, z_j}\rangle z_c-\langle{z_c, z_j}\rangle z_j), \qquad \dot{z}_j=\kappa_1(\langle{z_j, z_c}\rangle-\langle{z_c, z_j}\rangle)z_j.
\]
The first subsystem coincides with the Lohe sphere model in $\bbr^d$ so that similar emergent estimates were obtained. In contrast, the second subsystem is reminiscent of the Kuramoto model with heterogeneous frustrations for identical oscillators. Third, we provided a sufficient framework leading to the exponential aggregation estimate for the complex Lohe sphere model which combining the above two subsystems. Under the same framework in which exponential aggregation is guaranteed, we also show that the solution operator is uniformly $\ell^p$-stable with respect to initial data with $p \in [1, \infty)$. }
 
 There are some issues which were not discussed in this paper. For example, we have not discussed heterogeneous ensembles with distributed $\Omega_j$'s and detailed emerging patterns arising from initial data, effect of network structures and the kinetic mean-field limit from the Lohe hermitian sphere model. These interesting issues will be treated in a future work.


\end{document}